\documentclass[a4paper,UKenglish,cleveref, autoref, thm-restate]{lipics-v2021}

\usepackage[dvipsnames,usenames]{xcolor}

\RequirePackage{hyperref}
\hypersetup{colorlinks=true, urlcolor=Blue, citecolor=Green, linkcolor=BrickRed, breaklinks, unicode}

\usepackage[noadjust]{cite}

\usepackage{cite}
\usepackage[utf8]{inputenc}

\usepackage{amsthm}
\usepackage{amsmath}
\usepackage{amssymb}

\usepackage{xcolor}
\usepackage{todonotes}

\usepackage{thm-restate}

\newcommand{\R}{\mathbb{R}}
\newcommand{\Rp}{\R^+}

\newcommand{\DTW}{\mathsf{DTW}}

\newcommand{\F}{\mathsf{F}}
\newcommand{\AC}{\mathsf{AddConst}}
\newcommand{\AtR}{\mathsf{AddToRange}}
\newcommand{\AG}{\mathsf{AddGradient}}
\newcommand{\LLW}{\mathsf{LeftLinearWave}}
\newcommand{\RLW}{\mathsf{RightLinearWave}}

\newcommand{\assign}{\leftarrow}

\newcommand{\AssignOp}{\mathsf{Assign}}
\newcommand{\ShiftOp}{\mathsf{Shift}}

\newcommand{\dist}{\mathsf{dist}}

\newcommand{\local}{\text{short }}
\newcommand{\longrange}{\text{long }}

\newcommand{\addmin}{\text{Add-min }}

\newcommand{\Insert}{\mathsf{Insert}}
\newcommand{\Remove}{\mathsf{Remove}}
\newcommand{\Lookup}{\mathsf{Lookup}}
\newcommand{\Successor}{\mathsf{Successor}}
\newcommand{\Predecessor}{\mathsf{Predecessor}}

\newcommand{\Add}{\mathsf{AddConst}}
\newcommand{\Min}{\mathsf{Min}}

\newcommand{\nextGT}{\mathsf{nextGT}}
\newcommand{\prevLT}{\mathsf{prevLT}}

\renewcommand{\P}{\mathcal{P}}
\newcommand{\AP}{\P_{\mathsf{active}}}
\newcommand{\TAP}{\tilde{\P}_{\mathsf{active}}}

\newcommand{\DAL}{D_{\alpha}}
\newcommand{\Dbeta}{D_{\beta}}
\newcommand{\DGAM}{D_{\gamma}}

\newcommand{\TA}{\tilde{A}}
\newcommand{\tell}{\tilde{\ell}}
\newcommand{\TP}{\tilde{\P}}

\newcommand{\para}[1]{\vspace{-0.1in} \subparagraph*{#1}}

\newcommand{\B}{\mathcal{B}}

\newcommand{\ray}{r}

\newcommand{\rl}[1]{\rho_{#1}}

\newcommand{\flush}{\mathsf{flush}}

\newcommand{\ceil}[1]{\left\lceil{#1}\right\rceil}
\newcommand{\floor}[1]{\left\lfloor{#1}\right\rfloor}

\newcommand{\polylog}{\mathsf{polylog}}

\newtheorem{invariant}{Invariant}

\crefname{invariant}{Invariant}{Invariants}

\title{Near-Optimal Dynamic Time Warping on Run-Length Encoded Strings}
\author{Itai Boneh}{Reichman University and University of Haifa, Israel}{itai.bone@biu.ac.il}{}{}

\author{Shay Golan}{Reichman University and University of Haifa, Israel}{golansh1@macs.biu.ac.il}{https://orcid.org/0000-0001-8357-2802}{}

\author{Shay Mozes}{Reichman University, Israel}{smozes@idc.ac.il}{https://orcid.org/0000-0001-9262-1821}{}

\author{Oren Weimann}{University of Haifa, Israel}{oren@cs.haifa.ac.il}{https://orcid.org/0000-0002-4510-7552}{}

\authorrunning{I. Boneh, S. Golan, S. Mozes, and O. Weimann}

\Copyright{Itai Boneh, Shay Golan, Shay Mozes, and Oren Weimann}

\date{}

\funding{Israel Science Foundation grant 810/21.}



\category{} 

\relatedversion{}

\nolinenumbers 

\hideLIPIcs  

\begin{document}

\maketitle
\begin{abstract}
We give an $\tilde O(n^2)$ time algorithm for computing the exact Dynamic Time Warping  distance between two strings whose run-length encoding is of size at most $n$. This matches (up to log factors) the known (conditional) lower bound, and should be compared with the previous fastest $O(n^3)$ time exact algorithm and the $\tilde O(n^2)$ time approximation algorithm.  
\end{abstract}

\section{Introduction}
\label{section:introduction}

Dynamic Time Warping (DTW)~\cite{vintzyuk1968} is one of the most popular methods for comparing time-series (see e.g.~\cite{DBLP:journals/pr/Liao05,dtwapp1,dtwapp2,dtwapp3,dtwapp4,dtwapp6,WMDTSK13,ASW15,BLBLK17}). It is appealing in numerous applications such as bioinformatics, signature verification, and speech recognition, where two time-series can vary in speed but still be considered similar. For example, in speech recognition, DTW can detect similarities even if one person is talking faster than the other.  

To define DTW, recall that a run-length encoding $S=s_1^{\ell_1} s_2^{\ell_2} \cdots s_n^{\ell_n}$ of a string $S$ over an alphabet $\Sigma$ is a concise (length $n$) representation of the (length $N=\sum_i \ell_i$) string $S$. Here $s_i^{\ell_i}$ denotes a letter $s_i\in\Sigma$ repeated $\ell_i$ times.
For example, the string $S= aaaabbbaaaaa$ is encoded as $a^4 b^3 a^5$.
A string $S'=s_1^{\ell'_1} s_2^{\ell'_2} \cdots s_n^{\ell'_n}$ is a {\em time-warp} of string $S=s_1^{\ell_1} s_2^{\ell_2} \cdots s_n^{\ell_n}$ if  every $\ell'_i \ge \ell_i$.   

\begin{definition}[Dynamic Time Warping]
	
For a function $\delta : \Sigma^2 \rightarrow \Rp$, 
 the Dynamic Time Warping distance of two strings $S$ 
and $T$ over alphabet $\Sigma$ is defined as 
$$
\DTW(S,T) =  \min_{|S'| = |T'|} \sum_{i = 1}^{|S'|} \delta(S'[i], T'[i]),
$$ where $S'$ and $T'$ range over all time-warps of $S$ and $T$ respectively. 
\end{definition}

In 1968, Vintzyuk~\cite{vintzyuk1968} gave an $O(MN)$ time dynamic programming algorithm for computing the DTW of two strings $S$ and $T$ of lengths $N$ and $M$ respectively. His algorithm is one of the earliest uses of dynamic programming and is  taught today in basic algorithms courses and textbooks. Apart from logarithmic factor improvements~\cite{DBLP:conf/icalp/GoldS17}, the $O(MN)$ quadratic time complexity remains the fastest known and a strongly subquadratic-time  $O((MN)^{1-\varepsilon})$ algorithm is unlikely as it would refute the popular Strong Exponential
Time Hypothesis (SETH) \cite{DTWhard2,DTWhard}. 

The complexity of DTW in terms of $N$ and $M$ is thus well understood. Special cases of DTW are also well understood. These include DTW on binary strings \cite{kuszmaul2021binary, schaar2020faster}, approximation algorithms \cite{DBLP:conf/icalp/Kuszmaul19, DTWapprox1, DTWapprox2}, the large distance regime \cite{DBLP:conf/icalp/Kuszmaul19}, sparse inputs \cite{sparseDTW1,sparseDTW2,sparseDTW3},  
and reductions to other similarity measures
\cite{DBLP:conf/icalp/Kuszmaul19,sakai2020reduction, sakai2020faster}.
However, the complexity of DTW is not yet resolved in terms of $n$ and $m$ (the run-length encoding sizes of $S$ and $T$ respectively). Namely, in the (especially appealing) case where the strings contain long runs. 
The currently fastest algorithms are $O(Nm+Mn)$~\cite{rle1,DBLP:journals/corr/abs-1903-03003,DBLP:conf/icalp/Kuszmaul19} and $O(n^2m +m^2n)$~\cite{DBLP:journals/corr/abs-1903-03003}. In particular, an  $\tilde O(nm)$ time algorithm is only known to be possible if we are willing to settle for a $(1+\varepsilon)$-approximation~\cite{DBLP:conf/esa/XiK22}. It remained an open question whether it is possible to obtain an exact $\tilde O(nm)$ algorithm (which is optimal up to log factors). In this paper we answer this open question in the affirmative.

\para{Prior work on DTW.}   

The classical dynamic programming for DTW is as follows. Let 
 $\DTW(i,j) = \DTW(S[1\ldots i], T[1\ldots j])$, then $\DTW(0,0) = 0$, $\DTW(i,0)=\DTW(0,j) = \infty$ for every $i>0$ and $j>0$, and otherwise:
\begin{equation}\label{eq:dtw}
\DTW(i,j)=   \delta(S[i],T[j]) + \min\begin{cases} 
\DTW(i-1,j)
\\
\DTW(i,j-1) 
\\
\DTW(i-1,j-1)
\end{cases}
\end{equation}

The above dynamic programming is equivalent to a single-source shortest path (SSSP) computation in the following grid graph. We denote $[n] = \{1,2,\ldots,n\}$.
\begin{definition}[The Alignment Graph]
The alignment graph of $S$ and $T$ is a directed weighted graph $G$ with vertices $V = [0\ldots N] \times [0\ldots M]$.
Every vertex $(i,j)\in [N]\times[M]$ 
has three entering edges, all with weight $\delta(S[i],T[j])$: A vertical edge from  $(i-1,j)$, a horizontal edge from  $(i,j-1)$, and a diagonal edge from $(i-1,j-1)$. 
\end{definition}

We denote the distance from vertex $(0,0)$ to $(i,j)$ as $\dist(i,j)$.\footnote{Abusing notation, we will later also use $\dist((x,y), (x',y'))$ to denote the distance from vertex $(x,y)$ to vertex $(x',y')$.} 
 Clearly, $\DTW(i,j) = \dist(i,j) $. Therefore, $\DTW(S,T)= \dist(N,M)$ and can be computed in $O(MN)$ time by an SSSP algorithm (that explicitly computes the distances from $(0,0)$ to all the $O(MN)$ vertices of the graph). The way to beat $O(MN)$ is to only compute distances to a subset of vertices. 

Namely, partition the alignment graph into {\em blocks} where each block is the subgraph corresponding to a single run in $S$ and a single run in $T$. Then, proceed block-by-block and for each block compute its {\em output} (the last row and last column) given its {\em input} (the last row of the block above and the last column of the block to the left). Since blocks are highly regular (i.e., all edges inside a block have the same weight), it is not difficult to compute the output in time linear in the size of the output. Since the total size of all outputs (and all inputs) is $O(Nm+Mn)$, this leads to an overall   $O(Nm+Mn)$ time algorithm~\cite{rle1,DBLP:journals/corr/abs-1903-03003,DBLP:conf/icalp/Kuszmaul19}. 
 
 In order to go below $O(Nm+Mn)$, in~\cite{DBLP:journals/corr/abs-1903-03003} it was observed that we do not really need to compute the entire output. It suffices to compute only the intersection of the output with a set of $O(mn)$ diagonals. Specifically, each block contributes one diagonal starting in its top-left corner, so there are overall $O(mn)$ diagonals and each diagonal intersects with $O(m+n)$ blocks. This leads to an $O(n^2m +m^2n)$ time algorithm. 
Finally, in~\cite{DBLP:conf/esa/XiK22} it was shown that if we are willing to settle for a $(1+\varepsilon)$-approximation, then it suffices to compute only $\tilde O(1)$ output values. 

\para{Prior work on Edit distance.}
There are many similarities between DTW and the edit distance problem: (1)  
like DTW, edit distance can be computed
in $O(MN)$ time using the alignment graph \cite{vintzyuk1968,
  needleman1970general}. The only difference is in the edge-weights. (2) 
like DTW, edit distance has a lower bound prohibiting strongly
subquadratic time algorithms conditioned on the SETH \cite{DTWhard,
  DTWhard2,DBLP:conf/icalp/Kuszmaul19}, and (3) like DTW, edit distance can be computed in $O(Nm+Mn)$ time by proceeding block-by-block and computing the outputs from the inputs. However, unlike DTW, it is known how to compute the
edit distance of run-length encoded strings in $\tilde{O}(nm)$ time
\cite{rle1, rle2, rle3, rle4, rle5, rle6, rle7, rle8, rle9}. Specifically, Clifford et. al.~\cite{rle1} showed that the input and output of a block can be implicitly represented by a piecewise linear function, and, that the representation of the output can be computed in amortized $O(\polylog (mn))$ time from the representation of the input. This implies an $\tilde{O}(nm)$ time algorithm for edit distance. 

In \cite{DBLP:conf/esa/XiK22}, Xi and Kuszmaul  write about the prospects of obtaining an $\tilde{O}(nm)$ time algorithm for DTW: 
``{\em Such an
algorithm would finally unify edit distance and DTW in the
run-length-encoded setting}''. 

\para{Our result and techniques.}
 We present an $\tilde O(nm)$ time algorithm for DTW. This is optimal up to logarithmic factors under the SETH. Our algorithm is independent of the alphabet size $|\Sigma|$ and of the function $\delta$. In fact, $\delta$ need not even satisfy the triangle inequality. 
 
 We follow the approach for edit distance by Clifford et. al.~\cite{rle1} of representing and manipulating inputs and outputs with a piecewise-linear function. However, the manipulation is more challenging for several reasons which were highlighted by Xi and Kuszmaul \cite{DBLP:conf/esa/XiK22}: (1) unlike edit distance, DTW does not satisfy the triangle inequality. 
  (2) we are interested in arbitrary cost functions $\delta$ for DTW, whereas the $\tilde{O}(nm)$ algorithm for edit distance~\cite{rle1} works only for Levenshtein distance (when  $\delta(\cdot,\cdot) \in \{0,1\}$). 
  (3) in the standard setting (i.e. not the run-length encoded setting)  
edit
distance actually reduces to DTW
\cite{DBLP:conf/icalp/Kuszmaul19}. 

In \cref{sec:DTWtoRanges}, we show  that the required manipulation of inputs and outputs naturally reduces to $O(nm)$ operations on a data structure that, given an array $A$ of size $M+N$ initialized to all zeros, supports the following range operations:
 
\begin{definition}[Range Operations]\label{def:rangeops}
\
\begin{itemize}
    \item $\Lookup(i)$ - return $A[i]$.
    \item $\AC(i,j,c)$ -  for every $k\in [i\ldots j]$, set $A[k] \leftarrow A[k]+c$.
    \item $\AG(i,j,g)$ - for every $k\in [i\ldots j]$, set $A[k] \leftarrow A[k]+ k \cdot g$.
    \item $\LLW(i,j,\alpha)$ - for every $k\in [i\ldots j]$, set $A[k] \leftarrow \min_{t \in [i\ldots k]} \big(A[t] + (k-t)\alpha \big)$. 
    \item $\RLW(i,j,\alpha)$ - for every $k\in [i\ldots j]$, set  $A[k] \leftarrow \min_{t \in [k\ldots j]} \big(A[t] + (t-k)\alpha \big)$.
\end{itemize}

\end{definition}

In \cref{sec:IntervalDS}, we show our main technical contribution:  

\begin{theorem}\label{t:mainds}
Performing $s$ range operations of \cref{def:rangeops} can be done in amortized $O(\polylog (s))$ time per operation.
\end{theorem}

The proof of \cref{t:mainds} can be roughly described as follows: We represent the array $A$ by the line segments of the linear interpolation of $A$.
This way, the range operations of \cref{def:rangeops} translate to creating and deleting segments, changing their slopes, and shiftings segments up and down. For most operations, these changes apply to a single contiguous range of $A$ and are therefore quite simple to implement in polylog time. 
The difficult operations are $\LLW$ and $\RLW$. These operations may need to replace each of $\Omega(n)$ different sets of consecutive segments with a single new segment. We refer to the process of replacing  a set of consecutive segments with a single new segment as a {\em ray shooting} process.
Shooting each of these rays separately would be too costly.
More accurately, a ray shooting process that replaces many segments with a single one is not problematic since its cost can be charged to the decrease in the number of segments.
The challenge is in shooting rays that replace a single segment with another one, as this does not decrease the number of segments. 

Our main technical contribution is a sophisticated lazy approach for handling the problematic ray shooting processes. 
We study the structural properties of ray shooting processes, and characterize \longrange rays which we can afford to shoot explicitly, and \local rays, which we cannot. The structure we identify allows us to divide the segments representing $A$ into mega-segments, and keep track of a single pending \local ray in each mega-segment such that executing the pending ray shooting process in each mega-segment would result in the correct representation of the array $A$. While we cannot afford to actually carry out all of these pending ray shooting processes, we can afford to perform the process locally, e.g., in order to support $\Lookup$ for a specific element of $A$, or to facilitate the other range operations. 

One component of our lazy approach is a data structure (sometimes called {\em Segment tree beats} in programming olympiads) for the following problem: Maintain an array $A$ under lookup queries and two kinds of update: $\Add(i,j,c)$ - for every $k\in [i\ldots j]$ set $A[k] \leftarrow A[k]+c$, and $\Min(i,j,c)$ - for every $k\in [i\ldots j]$ set $A[k] \leftarrow \min\{A[k],c\}$. Though we are not aware of any official publication, it is known (see e.g.~\cite{beats}) that this problem can be solved in amortized polylog time. In \cref{sec:minadd} we show a different and {\em worst-case} polylog time solution.\footnote{We note that the solution in~\cite{beats} also supports range-sum queries and for such a conditional lower bound (from the Online Matrix-Vector Multiplication (OMV) problem) is  known~\cite{Private}. The lower bound implies that {\em worst-case} operations unlikely to be possible in $O(n^{1/2-\varepsilon})$ time. 
We are able to circumvent this lower bound because we only support lookups, but not range-sum queries. 
}

\section{DTW via Range Operations}
\label{sec:DTWtoRanges}

In this section we prove that \cref{t:mainds} implies an $\tilde O(nm)$ algorithm for DTW. Namely, that DTW reduces to efficiently supporting the range operations of \cref{def:rangeops}.

\para{Blocks in the alignment graph.}
Let $S[i_1 \ldots i_2]$ and $T[j_1 \ldots j_2]$ be the $i$'th run in $S$ and the $j$'th run in $T$ respectively. 
The block $B_{i,j}$ in the alignment graph is the set of vertices $(a,b)$ with $a\in [i_1 \ldots i_2]$ and $b\in [j_1 \ldots j_2]$.
All of the edges entering any vertex in block $B_{i,j}$ have the same weight $\delta(S[i_1],T[j_1])$, which we denote by $c_{B_{i,j}}$. 
We call the blocks $B_{i-1,j}$, $B_{i,j-1}$ and $B_{i-1,j-1}$ the {\em entering blocks} of $B_{i,j}$. The {\em input} of a block consists of all vertices belonging to the first row or first column of the block. The {\em output} of a block consists of all vertices belonging to the last row or last column of the block. The following structural lemma was also used implicitly in previous works (see formal proof in the appendix). 

\begin{lemma}
\label{lem:onlydiagonal}
Let $B$ be a block. 
\begin{itemize}
    \item If $(x,y),(x,y+1) \in B$ then there is a shortest path from $(0,0)$ to $(x+1,y+1)$ that does not visit $(x,y+1)$.
    \item If $(x,y),(x+1,y) \in B$ then there is a shortest path from $(0,0)$ to $(x+1,y+1)$ that does not visit $(x+1,y)$.
    \item If $(x,y),(x+1,y+1) \in B$ then there is a shortest path from $(0,0)$ to $(x+1,y+1)$ that goes through $(x,y)$.
\end{itemize}
\end{lemma}

\para{Frontiers in the alignment graph.}
Our algorithm for DTW processes all blocks in the alignment graph. At every step, the algorithm can processes any block $B$ as long as all its entering blocks have already been processed.
When block $B$ is processed, the algorithm computes $\dist(x,y)$ for every output vertex $(x,y)$ of $B$.
After processing block $B$, we say that the output vertices of $B$ are {\em resolved}.
At every step of the algorithm, the {\em frontier} is the set of resolved vertices with an outgoing edge to a block that was not yet processed. 
Observe that, at any given time in the execution of the algorithm, for every value $d \in [-N \ldots M]$, the frontier includes exactly one vertex $(x,y)$ such that $y-x = d$. 
At every step $t$ of the algorithm, we will maintain an array $\F_t[-N \ldots M]$ where $F_t[d] = \dist(x,y)$ such that vertex $(x,y)$ belongs to the current frontier and $y-x = d$. 
The main result of this section is the following lemma:
\begin{lemma}\label{l:outputsfromoutputs} 
 $F_{t+1}$ can be obtained by using $O(1)$ range operations (\cref{def:rangeops}) on $F_{t}$. 
\end{lemma}
  
  Before we prove \cref{l:outputsfromoutputs}, we prove that it implies our main result:
  
\begin{figure}[htb]
  \begin{center}
 \includegraphics[scale=0.8]{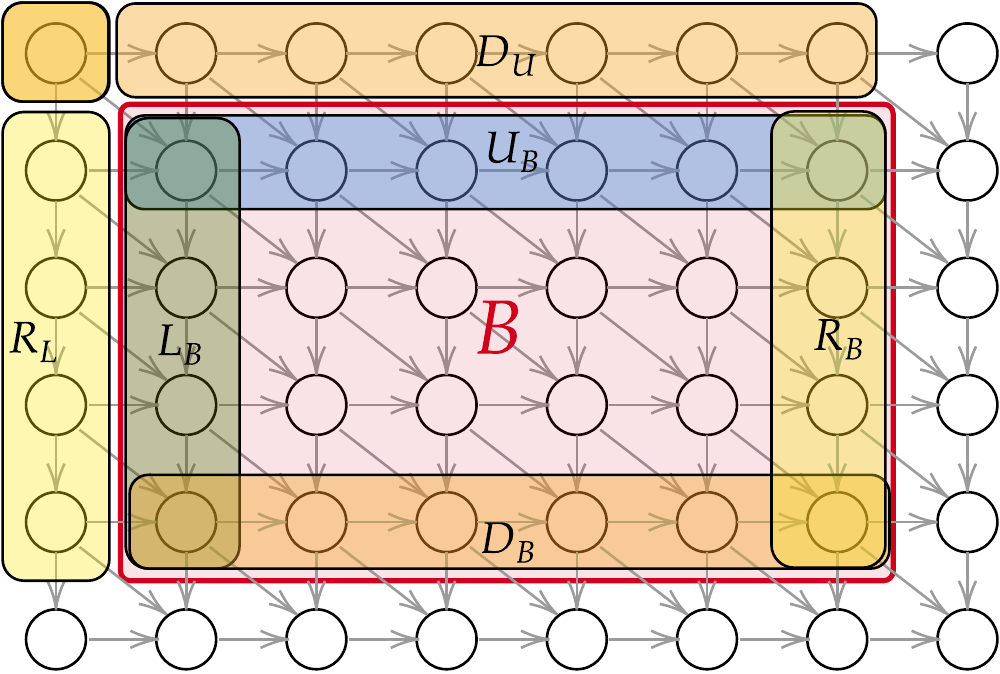} 
  \caption{A block $B$, its inputs $L_B\cup U_B$ and its outputs $D_B\cup R_B$. 
  The entering edges to $L_B$ are from $R_L$, the corner of $C$, and the leftmost node of $D_U$.
  The entering edges to $U_B$ are from $D_U$, the corner of $C$, and the topmost node of $R_L$.  \label{fig:BlockSketch}}
   \end{center}
\end{figure}

\begin{theorem}\label{t:maindtw}
The Dynamic Time Warping distance of two run-length encoded strings $S$ and $T$ with $n$ and $m$ runs respectively can be computed in $\Tilde{O}(nm)$ time.
\end{theorem}
\begin{proof}
We initialize the data structure of \cref{t:mainds} as an array of length $N+M+1$. 
We treat the indices of $A$ as if they are in $[-N \ldots M]$\footnote{When a gradient update $\AG(i,j,c)$ affects a value $A[k]$, we would like $A[k]$ to be increased by $k\cdot c$ with $k \in [-N \ldots M]$ being the 'simulated' index rather then the actual index $k+ N + 1$. This can be achieved by applying an additional operation $\AC(i,j, (-N-1) \cdot c)$.}.
Initially, the frontier consists of the vertices $(x,0)$ with $x\in [0 \ldots N]$ and $(0,y)$ with $y\in [M]$.
We start by turning $A$ into $\F_0$.
According to \cref{eq:dtw}, we need to set $A[0] = 0$ and $A[i] = \infty$ for $i\neq 0$.
This can be done by applying $\AC(1,M, \infty)$ and $\AC(-N,-1,\infty)$. 

The algorithm runs in $nm$ iterations.
At the beginning of iteration $t$, we have $A = \F_t$.
The algorithm picks any block $B$ whose entering blocks have already been processed, and applies $O(1)$ range operations (due to \cref{l:outputsfromoutputs}) on $A$ in order to obtain $A = \F_{t+1}$.
After the last iteration, it is guaranteed that the block $B_{n,m}$ has been processed. 
Therefore, $\F_{nm}[M-N]=\DTW(S,T)$.
Every iteration requires $O(1)$ range operations each in $O(\polylog(nm))$ time, so overall the algorithm performs $O(nm)$ operations in total $\Tilde{O}(nm)$ time.
\end{proof}

In the remainder of this section, we prove \cref{l:outputsfromoutputs}.        

\para{Overview.}
We obtain $\F_{t+1}$ from $F_{t}$ in two phases. Let $B$ be the block processed at step $t$, and suppose $B$ corresponds to runs $S[i_1\ldots i_2]$ and $T[j_1 \ldots j_2]$. Then $\F_t$ and $\F_{t+1}$ differ only in the range $[a\ldots b]$ where   $a= j_1-i_2$ and $b=j_2-i_1$ (see \cref{fig:pathsinblocks}).
In phase I we apply a sequence of range operations on $\F_{t}$ in order to obtain $\F$, which is defined to be identical to $\F_{t}$ except that the inputs of $B$ replace the corresponding outputs of $B$'s entering blocks.
 Formally,  
 $\F[d] = \F_t[d]$ for every $d \notin [a\ldots b]$, and for $d\in [a\ldots b]$, $\F[d] = \dist(x,y)$ where $(x,y)$ is the input node of $B$ with $y-x = d$. 
In phase II we apply another sequence of range operations on $\F$ to obtain $\F_{t+1}$, which is identical to $\F$ except that the inputs of $B$ are replaced by the outputs of $B$.  

The height of $B$ is denoted $h=i_2-i_1+1$ and the width of $B$ is $w=j_2-j_1+1$.
We denote the first row of $B$ as 
$U_B$, the first column of $B$  as 
$L_B$, the last row of $B$ as 
$D_B$, and the last column of $B$ as 
$R_B$ (see \cref{fig:BlockSketch}). 
We note that the input nodes of $B$ are $L_B\cup U_B$ and the output nodes are $D_B\cup R_B$.
We denote the entering blocks of $B$ as  $L=B_{i,j-1}$, $C=B_{i-1,j-1}$ and $U=B_{i-1,j}$.
We define $R_L$ as the last column of $L$, and $D_U$ as the last row of $U$.
Notice that the values of $\dist(x,y)$ for vertices of $R_L$ are stored in $\F_t[a-1\ldots a+h-2$], and the values of $\dist(x,y)$ for vertices of $D_U$ are stored in $\F_t[b-w+2\ldots b+1$] (see \cref{fig:pathsinblocks}).

\para{Phase I - computing $\F$ from $\F_t$.}
We begin by computing $\dist(i_1,j_1)$. 
Recall that $\dist(i_1-1,j_1-1)$ is stored in $\F_t[z]$ where $z=a+h-1=b-w+1=j_1-i_1$.
By \cref{eq:dtw}, we have $\dist(i_1,j_1)=c_B+\min\{\F_t[z-1],\F_t[z],\F_t[z+1]\}$.
Let $\Tilde{\F}_t$ be $\F_t$ with the assignment $\Tilde{\F}_t[z] \leftarrow \dist(i_1,j_1) - c_B$. 
The definition of $\Tilde{\F}$ is motivated by the following lemma.
\begin{lemma}\label{c:BotToTopFormula}
For every $k\in [z \ldots b]$, $\F[k] =c_B + \min_{i\in [z\ldots k]}\big( \Tilde\F_t[i] + (k-i)c_B \big)$.
\end{lemma}
\begin{proof}
For $k \in [z\ldots b]$, let $y =  k - z$.
Note that $\F[k]$ should be assigned $\dist(i_1,j_1 + y)$.
We prove the claim by induction on $y$.
For $y=0$, we need to prove that $\F[z]= \dist(i_1,j_1) = c_B +  \Tilde\F_t[z]$. This follows from the fact that $\Tilde{\F}_t[z] = \dist(i_1,j_1) - c_B$.
For the inductive step, we need to show that $\F[k] = \dist(i_1,j_1 + y) = \min_{i\in [z\ldots k]}\big( \Tilde\F_t[i] + (k-i)c_B \big) + c_B$.

By \cref{lem:onlydiagonal}, there are two options:
(i) The shortest path to $(i_1,j_1 + y)$ goes diagonally through $(i_1 - 1,j_1 + y -1)$. 
	Then, its length is $\dist(i_1 -1 , j_1 + y -1) + c_B= \Tilde\F_t[k] + c_B$ from the definition of $\Tilde\F_t$.
(ii) The shortest path to $(i_1,j_1 + y)$ goes horizontally through $(i_1, j_1 + y -1)$, and it follows that is length is $c_B + \F[k-1]$.
Then, by the induction hypothesis, the length of this path is
$$ \left (  c_B + \min_{i\in [z\ldots k-1]}\big( \Tilde\F_t[i] + ((k-1) - i)c_B \big)\right ) + c_B =  \min_{i\in [z\ldots k-1]}\big( \Tilde\F_t[i] + (k-i)c_B \big) + c_B. $$
Taking the minimum between (i) and (ii) yields the lemma. \qedhere
\end{proof}

It directly follows from \cref{c:BotToTopFormula} that $\F_t[z\ldots b]$ can be turned into $\F[z\ldots b]$ by applying the following range operations, in order:
\begin{enumerate}
    \item $\AC(z,z, \dist(i_1,j_1) - c_B - \F_t[z] )$.
    \item $\LLW(z,b,c_B)$.
    \item $\AC(z,b,c_B)$.
\end{enumerate}
The first operation turns $\F_t$ into $\Tilde\F_t$ and the other two operations turn $\Tilde\F_t[z\ldots b]$ into $\F$ by applying the formula given in \cref{c:BotToTopFormula}.
In a similar way, we can prove that $\F[a\ldots z]$ can be obtained from $\F_t$. This time,  
using $\RLW(i,j,c)$.

\para{Phase II - computing $\F_{t+1}$ from $\F$.} 

The following lemma will show that $\F_{t+1}$ can be obtained by applying $O(1)$ range operations on $\F$.
 
\begin{lemma}	
\label{cor:inputstooutputsformula}
Let $d_1 = j_1 - i_1$ and $d_2 = j_2 - i_2$. For every $d\in [a\ldots b]$:
\[
\F_{t+1}[d]=   \begin{cases} 
\F[d] + (d-a)c_B, &  \text{if } d\in [a\ldots \min(d_1,d_2)) 
\\
\F[d] + \min(w,h)c_B, &  \text{if } d\in [\min(d_1,d_2) \ldots \max(d_1,d_2)]
\\
\F[d] + (b-d)c_B, & \text{if } d\in (\max(d_1,d_2) \ldots b]
\end{cases}
\]
\end{lemma}
\begin{proof}
We begin with the following claim:	
\begin{claim}\label{o:pathviadiagonalinput}
Let $(x,y)$ be a vertex in block $B$ and let $(x',y')$ be the input vertex of $B$ with $y'-x' = y-x$, then $\dist(x,y) = \dist(x',y') + c_B \cdot (x-x')$. 
\end{claim}
\begin{proof}
We prove by induction on $d = x-x'$ that there is a shortest path from $(0,0)$ to $(x,y)$ that visits $(x',y')$ and then uses $x'-x$ consecutive diagonal edges. 
If $d = 0$, this holds trivially.
Otherwise, since $d \ge 1$, the vertices $(x,y)$ and $(x-1,y-1)$ are both in the same block $B$.
It follows from \cref{lem:onlydiagonal} that there is a shortest path to $(x,y)$ via $(x-1,y-1)$, which by the  induction hypothesis goes through $(x',y')$ and then uses only diagonal edges. 
\end{proof}

\begin{figure}[htpb]
  \begin{center}
 \includegraphics[scale=0.7]{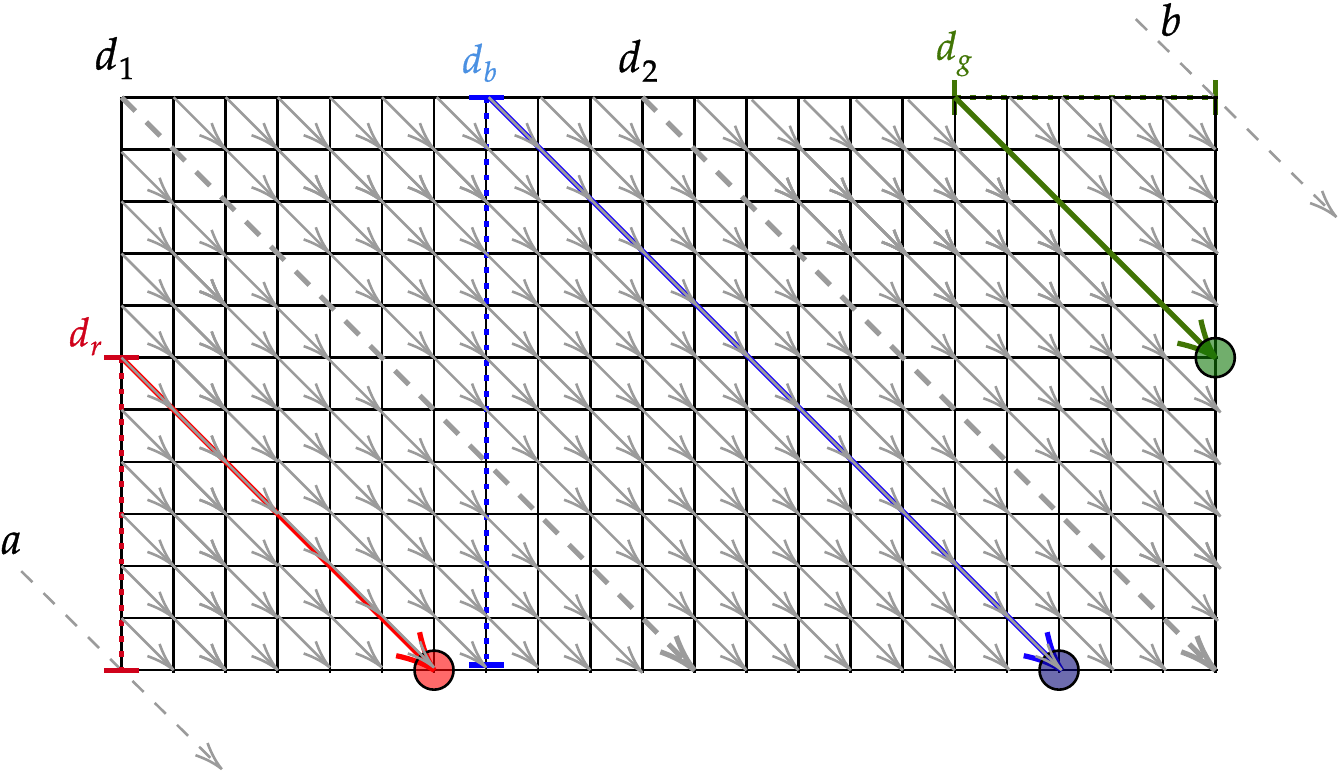} 
  \caption{ 
  A block $B$. The red vertex is an output vertex on a diagonal $d_r$ with $d_r \in [a\ldots d_1]$, and the number of diagonal steps from the matching input on the same diagonal to the red vertex is $d_r - a$. Similarly, the blue vertex is on a diagonal $d_b \in [d_1 \ldots d_2]$,  and the green vertex is on a diagonal $d_g \in [d_2 \ldots b]$.\label{fig:pathsinblocks}}
   \end{center}
\end{figure}

Let $(x,y)$ be an output of $B$ on diagonal $d=y-x$. Let $(x',y')$ be the input of $B$ on the same diagonal $d$. Therefore, the following holds (see \cref{fig:pathsinblocks}): 

\begin{enumerate}
    \item If $d \in [a \ldots \min(d_1,d_2))$, we have $x-x' = d - a$.
    \item If $d\in [\min(d_1,d_2) \ldots \max(d_1,d_2)]$, we have $x - x' = \min(w,h)$.
    \item If $d\in (\max(d_1,d_2) \ldots b]$, we have $x-x' = b-d$.
\end{enumerate}

\noindent Combined with \cref{o:pathviadiagonalinput}, this proves \cref{cor:inputstooutputsformula}.
\end{proof}

\noindent From \cref{cor:inputstooutputsformula}, $\F$ can be turned into $\F_{t+1}$ by applying: 
\begin{enumerate}
    \item $\AC(a, \min(d_1,d_2) - 1, -a \cdot c_B)$.
    \item $\AG(a, \min(d_1,d_2) - 1, c_B)$.
    \item $\AC(\min(d_1,d_2),\max(d_1,d_2),\min(w,h)c_B)$.
    \item $\AC(\max(d_1,d_2) +1, b, b \cdot c_B)$.
    \item $\AG(\max(d_1,d_2) + 1, b , -c_B)$.
\end{enumerate}

This concludes Phase II and the proof of \cref{l:outputsfromoutputs}. 
As for the parameters of the required range queries $d_1,d_2,w,h,a,b,c_B,$ and $z$ can all be calculated in advance for every block $B$ in $O(nm)$ time in a straightforward manner.

\section{Implementing the Range Operations}\label{sec:IntervalDS}

In this section we prove \cref{t:mainds}.
We view the array $A$ as a piecewise linear function. 
We associate with $A$ a set $\P = \{p_1 = (x_1,y_1) ,p_2 = (x_2,y_2), \ldots\}$ of points satisfying $A[x_i]=y_i$. The set $\P$ is uniquely defined by $A$ as the endpoints of the maximal linear segments of the linear interpolation of $A$. Note that the first point of $\P$ is always $(1,A[1])$ and the last point is $(n,A[n])$.\footnote{Here we use $n$ to denote the size of the array $A$.} 
Let $\ell_i(x)=\alpha_ix+\beta_i$ be the line segment between $p_i$ and $p_{i+1}$. 
Our representation will maintain the $\alpha_i$'s and $\beta_i$'s.
With this representation we can retrieve $A[x]$ for any $x\in[1,n]$ from $\alpha_i$ and $\beta_i$ where $x_i$ is predecessor of $x$ in the sequence $(x_1, x_2, \ldots)$.
Upon initialization, $A$ is represented as one linear segment, with $\alpha_1 =0$, and $\beta_1=0$.

We will use the following simple data structure.\footnote{The data structure can be implemented using a balanced search tree with a delta-representation (where the value of a node is represented by the sum of values of its ancestors), and having every node also store the minimal and maximal values in its subtree. See e.g. \cite{PlanarBook}.}

\begin{lemma}[Interval-add Data Structure]\label{l:IntervalAdd}
    There is a data structure supporting the following operations in $O(\log n)$ time per operation on a set of $n$ points with distinct first coordinates.
    \begin{itemize}
         \item $\Lookup(x)$ - return the second coordinate of the point with first coordinate $x$, if exists.
        \item $\Insert(x,y)$ - insert the point $(x,y)$.
        \item $\Remove(x)$ - remove the point with first coordinate $x$, if exists.
        \item $\AtR(i,j,c)$ - for every point (x,y) with $x\in[i\ldots j]$ set $y\assign y+c$. 
        \item $\nextGT(x,y)$ -\! return the point $p'=(x',y')$ with smallest $x'>x$ among points with $y'>y$.
        \item $\prevLT(x,y)$ - return the point $p'=(x',y')$ with largest $x'<x$ among points with $y'<y$. 
    \end{itemize}
\end{lemma}

\subsection{A Warmup algorithm}\label{sec:warmup}
We first present a naive and inefficient implementation of a range operations data structure.  
We maintain the sequence $\P$ in a predecessor/successor data structure over the sequence $(x_1, x_2, \ldots)$. With a slight abuse of notation we shall also use $\P$ to refer to this data structure. We maintain the $\alpha_i$'s and $\beta_i$'s using two Interval-add data structures $D_{\alpha}$ and $D_\beta$, respectively. 
The parameters $\alpha_i, \beta_i$ of the linear segment $\ell_i$ starting at $x_i$ are represented by points $(x_i,\alpha_i)$ in $D_\alpha$ and $(x_i, \beta_i)$ in $D_\beta$. In what follows, whenever we say we add a point $p=(x,y)$ to $\P$ we mean that $(x,y)$ is inserted into the predecessor/successor data structure $\P$, and that points with first coordinate $x$ are inserted into $D_\alpha$ and $D_\beta$, with their second coordinates appropriately set to reflect the parameters $\alpha,\beta$ of the segment ending at $p$ and the segment starting at $p$. This process requires $O(1)$ operations on $\P,D_\alpha$ and $D_\beta$. 
 
The effect of $\AC(i,j,c)$ (see \cref{fig:addconst}) is to break the segment containing $i$ into at most 3 linear segments (a prefix ending at $i-1$, a segment $[i-1,i]$, and a suffix starting at $i$), and similarly for the segment containing $j$.
Thus, to apply $\AC(i,j,c)$, we first replace the segments containing $i$ and $j$ with these $O(1)$ new segments by inserting or updating the endpoints of the segments in $\P,D_\alpha$, and $D_\beta$. 
We then invoke $\AtR(i,j,c)$ on $D_{\beta}$ to shift all segments between $i$ and $j$ by $c$.
Next, we set the parameters for the segment $[i-1,i]$ and for the segment $[j,j+1]$ by $O(1)$ additional calls to $\AtR$ on $D_\alpha$ and $D_\beta$. 
Finally, we check if any of the new segments we inserted has the same slope as its adjacent segments and, if so, we merge them into a single segment by removing their common point from $\P,D_\alpha$ and $D_\beta$. This guarantees that the set $\P$ we maintain is indeed the set $\P$ defined by $A$.
Supporting $\AG(i,j,g)$ is similar. The only difference is that we invoke  $\AtR(i,j,g)$ on $D_{\alpha}$ instead of on $D_{\beta}$ because the slope of the segments is shifted rather than their values.

\begin{figure}[!htb]
  \begin{center}
 \includegraphics[scale=0.8]{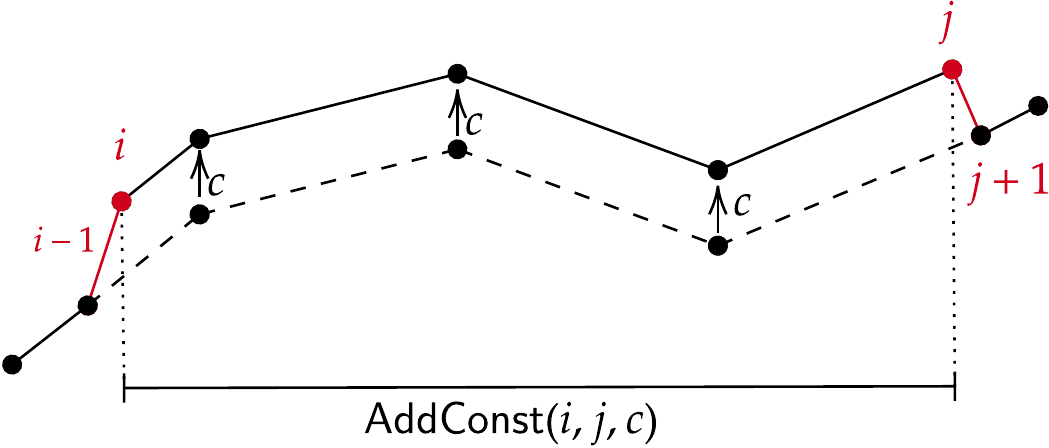} 
  \caption{An illustration of applying the $\AC(i,j,c)$ operation. The dashed line represents the segments before the operation. 
  After the operation, new points are created with $x$ coordinates $i-1,i,j$ and $j+1$ and the segments in $[i\ldots j]$ are shifted by $c$. 
  \label{fig:addconst} }
   \end{center}
\end{figure}

The challenge is thus in supporting $\LLW(i,j,\alpha)$. 
We first describe its effect and then describe how it is implemented. 
We assume without loss of generality that $i$ and $j$ are both endpoints of segments (otherwise we break the segments containing them into $O(1)$ segments as above).
Let $p_a=(i,A[i])$ and $p_{b+1}=(j,A[j])$ be the points corresponding to $i$ and $j$.
Thus, the segments contained within  $[i\ldots j]$ are $\ell_a, \ell_{a+1} \ldots \ell_b$.

If $\alpha_a\le \alpha $ then the segment $\ell_a$ is not affected by the linear wave. This is because 
for every  $k\in[x_a\ldots x_{a+1}]$, the linear wave assigns $A[k] \leftarrow \min_{x_a\le t\le k}(A[t] + (k-t)\alpha ) = $
$$ = \!\!\!
 \min_{x_a\le t\le k}(A[k]-(k-t)\alpha_a + (k-t)\alpha ) =\!\!\!
\min_{x_a\le t\le k}(A[k] +(k-t)(\alpha-\alpha_a)) = A[k].$$ 

Let $z\in[a\ldots b]$ be the minimum index such that $\alpha_z>\alpha$.
By the same reasoning, none of the segments $\ell_a,\ell_{a+1}, \ldots ,\ell_{z-1}$ is affected by the linear wave.
Let $\ray_z(x)$ be the (positive) ray with slope\footnote{Note that in \cref{fig:rayshootingexample} and in all subsequent figures we indicate the slope $\alpha$ of the ray $\ray_z$ by drawing an angle $\alpha$ between the ray and the positive direction of the $x$-axis. However, formally $\alpha$ is the slope of the ray, not the indicated angle.} $\alpha$ starting at $p_z$. Since $\alpha_z>\alpha$, the ray $\ray_z$ is below the linear segment $\ell_z$. 
Hence, the segment $\ell_z$ starting at $p_z$ is affected by the linear wave; its slope changes from $\alpha_z$ to $\alpha$, and it extends beyond $x_{z+1}$ as long as $A[x] \geq \ray_z(x)$.   
We describe this effect of $\LLW$ by a \emph{ray shooting} process from $p_z$ (See \cref{fig:rayshootingexample}).
This process identifies the new endpoint $p'$ of $\ell_z$, and removes all the existing segments between $p_z$ and $p'$, as follows.

Let $z'\in[z+1..b+1]$ be the minimum index with $y_{z'} <  \ray_z(x_{z'})$, i.e. the first point in $\P$ that lies strictly below the ray $\ray_z$.
Let $p^*=(x^*,y^*)$ be  the intersection point of the ray $\ray_z$ with $\ell_{z'-1}$ (if $z'$ does not exist, then $p^*=p_b$). 
The new endpoint of $\ell_z$ is the point $p' = (x',y') = (\floor{x^*},\ray_z(\floor{x^*})$, and it replaces all the points $p_w$ for $w\in (z\ldots z')$.
If $x^*$ is not an integer (or if $z'$ does not exist) then a new segment is formed between $p'$ and $p''=(x'+1,A[x'+1])$.  

\begin{figure}[!htb]
  \begin{center}
 \includegraphics[scale=0.8]{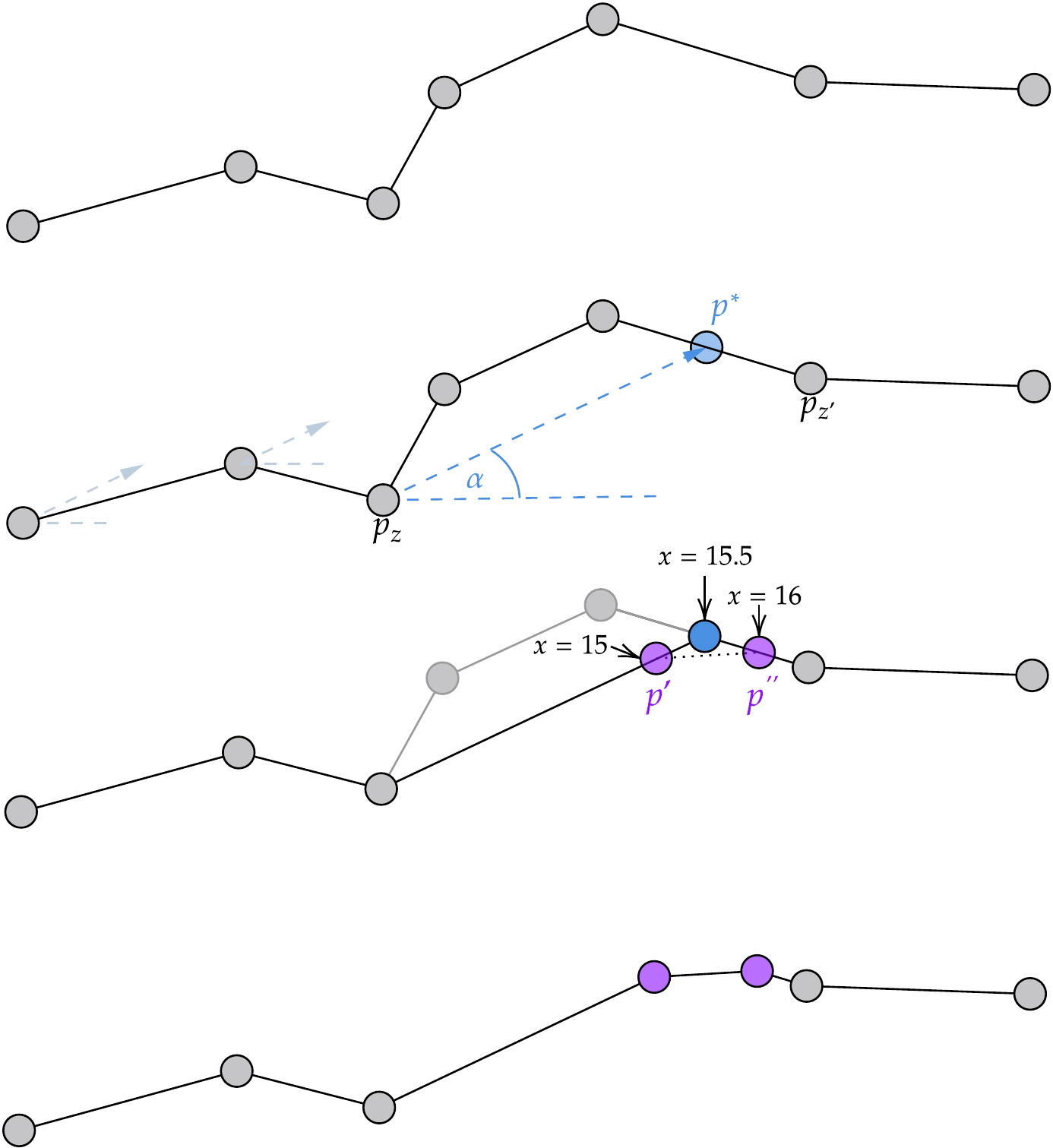} 
  \caption{
The effect of $\LLW(i,j,\alpha)$. 
The segments before $p_z$ are not affected. The segments between $p_z$ and $p_{z'}$ are affected. Namely, a ray $\ray_z$ with slope $\alpha$ (in dashed blue) is shot from $p_z$ and intersects at point  $p^*=(x^*,y^*)$.
The new endpoint of $\ell_z$ becomes $p'$ and all the segments between $p_z$ and $p'$ are removed. 
Since $x^*=15.5$ is not an integer, a new segment is formed between $p'$ (with $x$ coordinate 15) and $p''$ (with $x$ coordinate 16).
\label{fig:rayshootingexample}}
   \end{center}
\end{figure}

The effect of $\LLW(i,j,\alpha)$ on the remaining part of $A$, namely on $A[ x_{z'} \ldots j]$ is analyzed in the same way as above, this time starting from $p_{z'}$ instead of from $p_a$. 
In particular, the prefix of segments with slopes less than $\alpha$ is not  affected, and a ray with slope $\alpha$ is shot from the next $p_{w}$ with $\alpha_w > \alpha$, and so on.
In the appendix (\cref{lem:ray}) we formally prove that the above characterization indeed represents the new values of $A[i \ldots j]$.

We now describe a naive, non-efficient implementation of $\LLW(i,j,\alpha)$ according to the description above.  
Recall that $i$ and $j$ are assumed to be endpoints $p_a$ and $p_{b+1}$ of segments. 
We begin by finding the first $p_z$ with $x_z\in[i\ldots j]$ and $\alpha_z>\alpha$ by querying $D_{\alpha}.\nextGT(i,j,\alpha)$. A ray shooting process is then performed from $p_z$ (if $p_z$ exists) as follows:
Recall that $\ray_z(x)$ denotes the positive ray with slope $\alpha$ shot from $p_z$.
We scan the successor points of $p_z$ one by one in order, and for every point $p_w$ we check whether the ray $\ray_z(x_w)\leq y_w$.
If so, $p_w$ is removed by  removing $x_w$ from $\P$, $D_{\beta}$ and $D_{\alpha}$.
Otherwise, we compute $p^*=(x^*,y^*)$, the intersection point of $\ray_z$ and $\ell_{w-1}$, and from it the points $p' = (x',y')=(\floor{x^*},\ray_z(\floor{x^*})$ and, if $x^*$ is not an integer, also $p''=(\ceil{x^*},\ell_{w-1}(\ceil{x^*}))$.
Then, we insert the new points $p'$ and $p''$ just before $p_w$, as discussed above for $\AC$.
The scanning then continues with another $\nextGT$ query from $p_w$, and so on.
If, at the end of the process, the last point $p_b$ is removed since it is above some $\ray_z$, we insert a new point $(x_b,\ray_z(x_b))$.

We now analyze the time complexity of this naive implementation. Each $\AC$ and $\AG$ operation requires $O(1)$ operations on the Interval-add data structures, and therefore takes $O(\polylog(|\P|)$ time per operation, with $|\P|$ being the cardinality of $\P$ when the operation is applied. 

Regarding $\LLW$ operations, one might hope that the cost of each ray shooting process can be charged to the removal of points from $\P$ during the process. However, each ray shooting process might also add up to two new points, which might result in the size of $\P$ increasing. 
Indeed, a $\LLW$ operation may give rise to many such ray shooting processes, and hence may significantly increase the size of $\P$ and take too much time. This is the main technical challenge we need to address.

The idea is to  distinguish between {\em \longrange}ray shootings for which we can globally charge the new insertions, and {\em \local}ray shootings for which we cannot. We handle the \longrange rays as in the naive solution and devise a separate lazy mechanism that delays the application of all the \local rays stemming from a single $\LLW$ operation using a constant number of updates to a separate data structure that keeps track of the delayed rays. 

\para{Symmetry of $\RLW$.}
The discussion so far was focused on the $\LLW$ operation.
We note that the analysis of $\RLW$ is symmetric.
In particular, the execution of $\RLW(i,j,\alpha)$ can be described as a sequence of ray shootings with \emph{negative} rays.
The first point from which a ray is shot is $p_z$ with largest $z\in[a\ldots b]$ such that $\alpha_{z-1}<-\alpha$ ($p_z$ is found using $D_{\alpha}.\prevLT$). 
Note that the condition for starting a ray shooting process for $\RLW$ is on $\alpha_{z-1}$ rather than $\alpha_z$ since the slope of the segment to the left of $p_z$ is $\alpha_{z-1}$.
To simplify the presentation, we will keep describing only $\LLW$, and will comment at the very end about the minor adjustments required to also handled the symmetric $\RLW$. 

\subsection{Active and Passive Points, Long and Short Rays}
On our way to formally define \longrange rays and short rays we first observe that ray shootings only occur at points where slopes increase. We call such points {\em active} points.
\begin{definition}[Active and Passive points]\label{def:activepoints}
A point $p_z$ in $\P$ is called \emph{active} if $z\in\{1,|\P|\}$ or $\alpha_z>\alpha_{z-1}$.
A point that is not active, is called \emph{passive}.
We denote the sets of active points by $\AP$.
\end{definition}

\begin{lemma}\label{lem:activerays}
Ray shootings stemming from $\LLW(i,j,\alpha)$ occur either at point $p_a=(i,A[i])$ or at active points. 
   
\end{lemma}
\begin{proof}
Assume to the contrary that a ray shooting process starts at a passive point $p_z \neq p_a$.
 If $p_z$ is the first point where a ray shooting starts, then $z$ is the minimal index in $[a\ldots b]$ with $\alpha_z > \alpha$. But since $p_z$ is passive, we have $\alpha < \alpha_{z} \le \alpha_{z-1}$, contradicting the minimality of $z$ (note that $p_{z} \neq p_a$ so $z-1 \in [a\ldots b]$).

Otherwise, let $p_q$ be the last point before $p_z$ from which a ray shooting process occurred.
Let $p_{q'}$ be the first point below the ray shot from $p_q$. Since $p_z$ is the next point from which a ray is shot, $z$ is the first point in $[q \ldots b]$ with $\alpha_z \ge \alpha$.
Since $p_z$ is passive, we have $\alpha < \alpha_{z} \le \alpha_{z-1}$.
If $z \neq q'$, we have $z-1 \in [q'\ldots b]$, a contradiction to the minimality of $z$.
Otherwise, $p_z = p_{q'}$ is the first point below the ray with slope $\alpha$ shot from $p_q$. 
It follows that $p_{z-1}$ is above the ray, and $\alpha_{z-1} > \alpha$.
It must be the case that $p_{q'}$ is above the ray, a contradiction.
\end{proof}

Similarly, we provide a proof in \cref{sec:appendix} for the following symmetric claim regarding $\RLW$
\begin{lemma}\label{lem:activeraysright}
Ray shootings stemming from $\RLW(i,j,\alpha)$ occur either at point $p_b=(j,A[j])$ or at active points.
\end{lemma}

Let $\AP = q_1,q_2 \ldots$ be the restriction of the sequence $\P$ to the active points. 
We can think of the active points as defining a piecewise linear function whose segments are a coarsening of the segments of $A$. We refer to these segments as {\em mega-segments}. 
Let $\gamma_z$ denote the slope of the mega-segment whose endpoints are $q_z$ and $q_{z+1}$. 
The following lemma asserts that the segments of $A$ are never below their corresponding mega-segments, and that the slope of a segment starting at an active point is never smaller than the slope of the mega-segment starting at the same point. 

\begin{lemma} \label{lem:above-mega}
Let $q_z = p_w$ and $q_{z+1} = p_{w'}$ be two consecutive active points.
For every $k \in [w \ldots w']$, the passive point $p_k$ is not below the mega-segment connecting $q_z$ and $q_{z+1}$.
Furthermore, $\alpha_w \ge \gamma_z$.
\end{lemma}
\begin{proof}
(See \cref{fig:abovemegaproof}) Assume by contradiction that there is a point below the mega-segment, and let $k'\in (w\ldots w')$ be the smallest index of such a point.
Since $p_{k'-1}$ is not below the mega-segment and $p_{k'}$ is below the mega-segment, we must have $\alpha_{k'-1} < \gamma_z$.
Moreover, since the points $p_k$ with $k \in [k' \ldots w')$ are passive, the slopes are non-increasing and therefore every $\alpha_k \leq \gamma_z$.
This means that all these points and in particular $p_{w'}$ are below the mega-segment. In contradiction to $p_{w'}$ lieing on the mega-segment.
\end{proof}

\begin{figure}[htb]
  \begin{center}
  \includegraphics[scale=0.8]{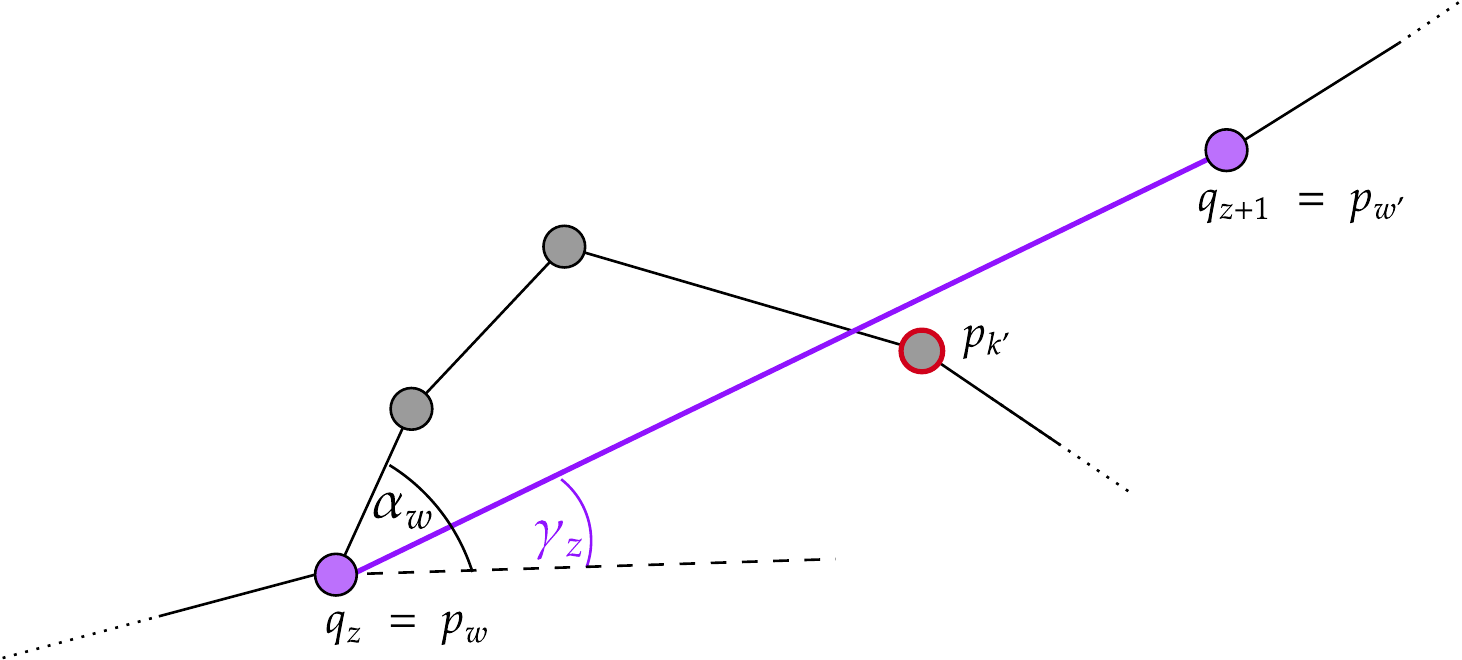} 
  \caption{The impossible configuration in  \cref{lem:above-mega}. 
  The points $q_z$ and $q_{z+1}$ are represented by the purple points and the mega-segment connecting them is represented by a thick purple line.
  The first point $p_{k'}$ below the segment is marked with red stroke. Since the points strictly within the mega-segment are passive, the points following $p_{k'}$ within the mega-segment (and in particular $q_{z+1}$) must remain below the mega-segment. \label{fig:abovemegaproof}}
   \end{center}
\end{figure}
We next show that if a ray shooting process starts at an active point $q_z$ with $\gamma_z< \alpha$ then the process ends before $q_{z+1}$, and the only affected points are the passive points between $q_z$ and $q_{z+1}$. On the other hand, if $\gamma_z\ge \alpha$ then as a result of the process $q_{z+1}$ ceases to be an active point, so $|\AP|$ decreases. 
\begin{lemma}\label{lem:rayshootinggamma}
    Consider a ray shooting process starting from point $p_w = q_z \in \AP$ during the application of $\LLW(i,j,\alpha)$. Let $q_{z+1} = p_{w'}$.
    \begin{enumerate}
    \item No new active points $p=(x,y)$ with $x \neq j$ are created in this process.
    \item If $\gamma_z<\alpha$ then the points that are deleted by this process are the (passive) points  $p_{k}$ with $k\in [w+1 \ldots r]$ for some $w < r < w'$.
    No other points between $p_w$ and $p_{w'}$ are deleted by $\LLW(i,j,\alpha)$. 
    \item If $\gamma_z \ge \alpha$ then $q_{z+1}$ is either deleted or becomes passive.
    \end{enumerate}
\end{lemma}
\begin{proof}
Let $\ell$ be the ray starting from $p_w=q_z$. Assume the process terminates by finding the first point $p^* = (x^*,y^*)$ below $\ell$ (the only process that does not end this way is the one that ends by reaching $(j,A[j])$).
The ray shooting process adds at most two new points $p'$ and $p''$ with decreasing slopes, so no new active points are created by the process.
The slope of $p'$ is decreasing because the segment entering $p'$ is (a sub-segment of) $\ell$ and the segment leaving $p'$ is to a point below $\ell$.
The slope of $p''$ is decreasing because the line segment entering $p''$ is a line from $p'$ (a point on $\ell$) and the line segment leaving $p''$ is to the suffix of a line segment below $\ell$.
 
Consider the case $\gamma_z < \alpha$. Then $q_{z+1}$ is below the ray with slope $\alpha$ starting at $q_z$. 
Hence the ray shooting process terminates at a point after $p_{w+1}$ and before $q_{z+1}$. Since no active points are created, the next ray will be shot from $q_{z+1}$ or later, so no other points between $q_z$ and $q_{z+1}$ are deleted by $\LLW(i,j,\alpha)$. 

Now consider the case $\gamma_z > \alpha$. Then the mega-segment between $q_z$ and $q_{z+1}$ is above the ray with slope $\alpha$ shot from $q_z$. By \cref{lem:above-mega}, all the (passive) points between $q_z$ and $q_{z+1}$ are also above this ray. Hence $q_{z+1}$ is deleted by the ray shooting process.

Finally, consider the case $ \gamma_z=\alpha$. Then the mega-segment between $q_z$ and $q_{z+1}$ coincides with the ray with slope $\alpha$ shot from $q_z$. By \cref{lem:above-mega}, all the (passive) points between $q_z$ and $q_{z+1}$ will be deleted by the ray shooting process. Let $w'$ be such that $q_{z+1}=p_{w'}$. If $\alpha_{w'} \geq \alpha$ then $q_{z+1}$ will be deleted by the process. Otherwise, $\alpha_{w'} < \alpha$, so the ray shooting process terminates at $q_{z+1}$. Since all the passive points between $q_z$ and $q_{z+1}$ were deleted, $q_z$ and $q_{z+1}$ become consecutive in $\P$, and the slope of the corresponding segment is $\gamma_z = \alpha$. But the slope of the segment starting at $q_{z+1}$ is $\alpha_{w'} < \alpha$, so $q_{z+1}$ becomes passive.
\end{proof}

We call rays with $\gamma_z > \alpha$ \emph{\longrange}rays, and those with $\gamma_z \le \alpha$ \emph{\local}rays. Since \longrange rays decrease $\AP$ we can handle them explicitly as in the warmup, charging the deletion of passive points during the process to their creation, and charging the insertion of the at most two passive points at the end of the process to the decrease in $|\AP|$. The \local rays, which do not decrease $|\AP|$, will be handled lazily. Namely, instead of explicitly shooting a \local ray in the mega-segment starting at an active point $q_z$, we only store the slope of the ray and postpone its execution until it is required (e.g., by a $\Lookup$ operation). Note that subsequent \local rays shot in this mega-segment may further change the stored slope, and subsequent \longrange rays may also affect it. We explain this in detail next. 

\subsection{The Data Structure}
Since our data structure is lazy, the sequence of points it maintains will be different than the sequence $\P$ that would have been maintained had we used the warmup algorithm from \cref{sec:warmup}. We will therefore use $\TP$ to denote the set of points actually maintained by the data structure. The points $\TP$ define linear segments $\tell_i(x)$ in the usual way. For $x \in [1,n]$ we denote by $\TA[x]$ the value $\tell_i(x)$, where $\tell_i$ is the segment containing $x$.
We stress that our algorithm does not maintain $\P$. However, for the sake of description and analysis only we shall keep referring to the original $\P$, and array $A$.
The definition of active and passive points, of the slopes $\gamma$ of mega-segments, and of \local and \longrange rays are now with respect to the slopes of the $\tell_i$'s.\footnote{It would have been more accurate to use $\tilde \alpha, \tilde \beta$, and $\tilde \gamma$, but this would be too cumbersome, so we stick to using $\alpha, \beta, \gamma$.} 
However, we shall maintain that the set of active points with respect to $\P$ and $\AP$ is the same:
\begin{invariant}\label{inv:TAP-AP}
$\TAP = \AP$.	
\end{invariant}

Following \cref{sec:warmup}, we maintain $\TP$ in a predecessor/successor data structure, as well as the Interval-add data structures $D_\alpha$ and $D_\beta$ representing the parameters of the linear segments $\tell_i(x)$ defined by the points of $\TP$. 
By implementing $\AC, \AG$ and \longrange ray shootings similarly to \cref{sec:warmup} (the exact details will be spelled out below),  we shall maintain the invariant that this part of the data structure correctly represents the values of active points.\footnote{See  \cref{inv:lazyrs} and the note following it.}  

We maintain the set of active points $\TAP=(q_1, q_2, \ldots)$ using a predecessor/successor structure on their $x$-coordinates. 
For each $q_z \in \TAP$, we maintain the slope $\gamma_z$ of the mega-segment starting at $q_z$ in an Interval-add data structure $D_\gamma$. In addition, we maintain a pending \local ray $\ray_z$ with slope $\rl{z}$ passing through $q_z$ (see \cref{fig:MegaSegmentData}) by maintaining $\rl{z}$ in a data structure $D_\rho$. This data structure, which we call the \addmin data structure is summarized below and described in detail in \cref{sec:minadd}. 

\begin{lemma}[\addmin Data Structure]\label{lem:addminds}
      There exists a data structure supporting the following operations in $O(\polylog n)$ time on a set of points $S$. 
    \begin{enumerate}
        \item $\Insert(x,y)$ - insert the point $(x,y)$ to $S$.
        \item $\Remove(x)$ - remove the a point $p=(x,y)$ from $S$, if such a point exists.
        \item $\Lookup(x)$ - Return $y$ such that $p=(x,y)$ is in $S$, or report that there is no such point.
        \item $\AtR(i,j,c)$ - for every $p=(x,y)\in S$ with $x\in[i\ldots j]$ set $y\assign y +c$. 
        \item $\Min(i,j,c)$ - for every $p=(x,y)\in S$ with $x\in[i\ldots j]$ set $y \assign \min (y ,c)$. 
    \end{enumerate}
    \end{lemma}
    
Note that storing $\rl{z}$ suffices to compute $\ray_z(x)$ since the active point $q_z$ that determines the free coefficient of $\ray_z$ is correctly represented by $D_\alpha$ and $D_\beta$.  
We shall show that storing a single pending ray suffices to represent all the pending changes in a mega-segment. This property will rely on maintaining the following invariant.
\begin{invariant}\label{inv:rr_rl}
    For every active point $q_z$ we have $\rl{z}>\gamma_z$. (Recall that $\gamma_z$ is the slope of the mega-segment connecting $q_z$ and $q_{z+1}$.)
\end{invariant}

The idea is that with this representation, for any $x$, the value of $A[x]$ is given by the minimum of the value $\TA[x] = \ell_w(x)$ of the segment of $\TP$ containing $x$, and the value $\ray_z(x)$ of the pending \local ray for the mega-segment containing $x$. This is captured by the following main invariant maintained by the data structure.

\begin{invariant}\label{inv:lazyrs}
Let $x \in [1,n]$, and let $p_w$ and $q_z$ be the predecessor of $x$ in $\TP$ and in $\TAP$, respectively. 
It holds that $A[x] = \min (\tell_w(x), \ray_z(x))$. Furthermore, if $p=(x,A[x])$ is an active point in $\P$, then $A[x] = \TA[x]$.
\end{invariant}

\begin{figure}[htb]
  \begin{center}
 \includegraphics[scale=0.8]{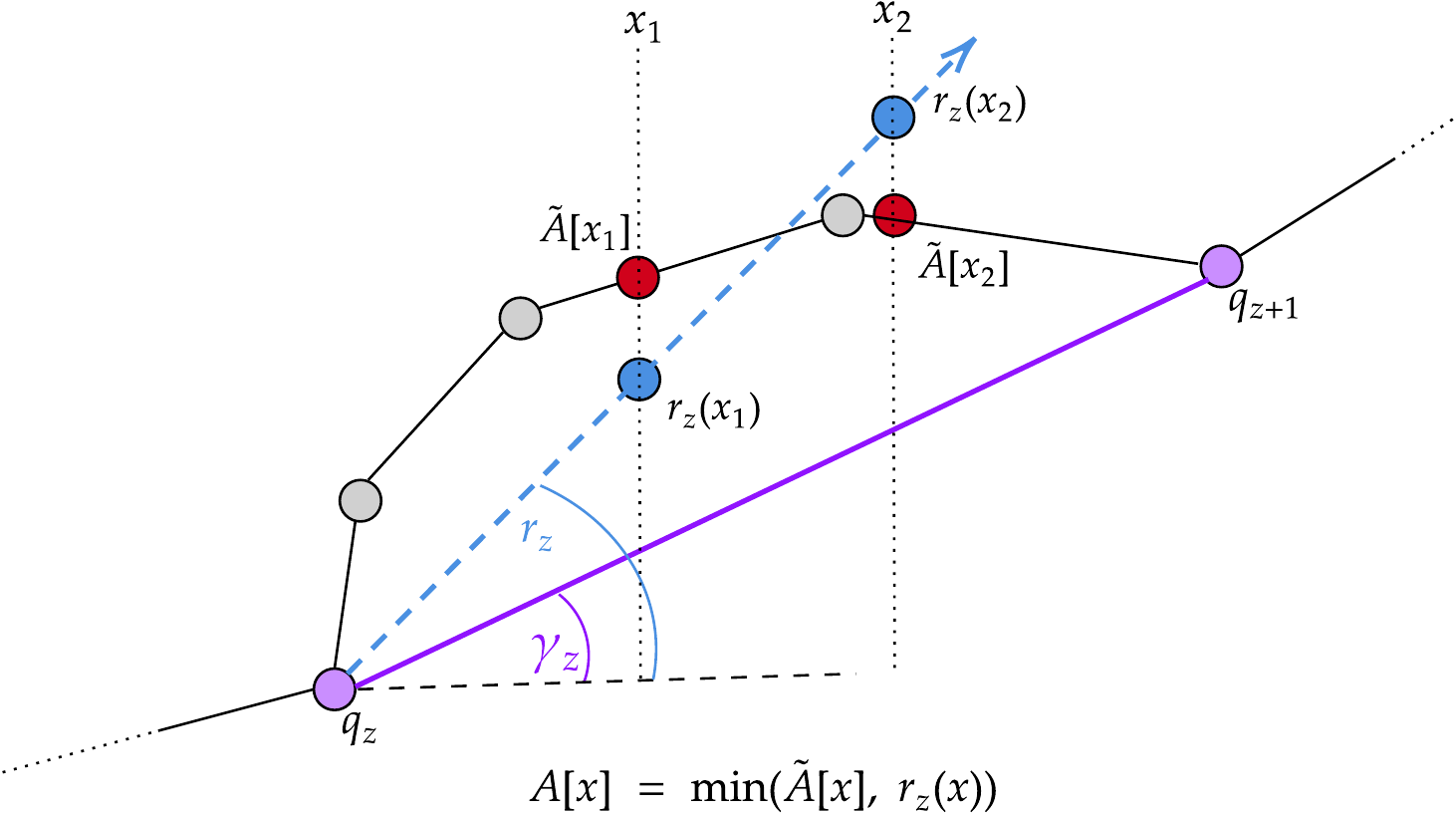} 
  \caption{An illustration of the data stored for a mega-segment between two consecutive active points $q_z$ and $q_{z+1}$ (purple points).
  The slope $\gamma_z$ is the slope of the mega-segment. 
  The slope $\rl{z} > \gamma_z$ stored in $q_z$ represents a pending ray $\ray_{z}$ (dashed blue) that should be shot from $q_z$.
  The value of $A[x]$ is the minimum between $\ray_{z}(x)$ (a blue point) and the $\TA[x]$ (a red point), the value of the piece-wise linear function defined by $\TP$ (in grey). \label{fig:MegaSegmentData} 
  }
   \end{center}
\end{figure}

Note that the first part of \cref{inv:lazyrs}, together with \cref{inv:TAP-AP} implies the second part of \cref{inv:lazyrs}.
This is because the predecessor of $x$ for an active point $p=(x,A[x])$ in $\AP$ is itself. Since $\AP=\TAP$ we have that $p\in \TAP$ is the predecessor of $x$ in $\TAP$ as well.
By definition $r_z$ goes through $p = q_z$, so $r_z(x) = \TA[x]$, and $\TA[x] = \tell_w(x)$ by definition. Hence, when proving that the invariants are maintained, we will not need to explicitly establish the second statement in \cref{inv:lazyrs}.

Initially, $\TP = \P = \{(1,0), (|A|,0)\}$, and 
$\rl{1} = \rl{2}  = \infty$. 
Indeed, $A[x] = \min(\TA[x], \ray_1(x)) = \min(0,\infty) = 0$ and \cref{inv:lazyrs} is satisfied.
It remains to specify  the implementation of the various operations supported by the data structure, to prove that the invariants are  maintained, and to analyze the running times.

\para{The $\flush$ Operation.}
We first describe a service operation $\flush(q_z)$ which explicitly shoots the pending \local ray in the mega-segment starting at the active point $q_z$. 
It will be useful to invoke $\flush$ before serving  $\Lookup$ operations, but also when serving the other operations in order to guarantee that the lazy implementation properly follows the explicit implementation in the warmup. This is particularly important in operations which may create $O(1)$ new active points and thus change the partition into mega-segments, but is also useful to streamline the proof of correctness. 
Recall that the reason we avoided shooting local rays in the first place was that there could be many of them, and we could not afford to pay for the possible creation of $O(1)$ new passive points at the end of each of them. We can afford, however, to perform $O(1)$ $\flush$ operations before each $\Lookup, \AC $ or $\AG$ operation, because the cost of adding the $O(1)$ new points can be charged to the operation itself. 

A $\flush$ of $q_z = p_w \in \AP$ is performed as follows.
Starting from $p_{w+1}$, we scan the points in $\TP$.
When scanning $p=(x,y)$, we compare $y$ and $\ray_z$. 
If $\ray_z(x) \le y$, we remove $p$ from $\TP$.
Otherwise, the scan halts. 
Let $p_{\mathsf{end}}$ be the point on which the scan halts.
If no point was deleted throughout the scan, we set $\rl{z} = \infty$ and terminate.
Otherwise, let $p_{\mathsf{del}}$ be the last point deleted by the scan.
We compute the intersection $p^*=(x^*,y^*)$ of $\ray_z$ and the line $\tell$ between $p_{\mathsf{del}}$ and $p_{\mathsf{end}}$.
Finally, we insert $p'=(\floor{x^*}, \ray_z(\floor{x^*}))$ and $p''= (\ceil{x^*}, \tell(\ceil{x^*}))$ to $\TP$ (as in the warmup algorithm of \cref{sec:warmup}), update $D_\alpha$ and $D_\beta$ with the new parameters of the segments ending and starting at $p'$ or at $p''$, and set $\rl{z} = \infty$.

\begin{lemma}\label{lem:flush}
Applying $\flush$ to an active point $q_z \in \AP$ preserves \cref{inv:rr_rl,inv:lazyrs,inv:TAP-AP}.
Furthermore, it guarantees that 
the restriction of $\P$ and $\TP$ to the (passive) points between $q_z$ and $q_{z+1}$
is identical,
and that for every $x\in[x_z \ldots x_{z+1}]$, $A[x]=\TA[x]$.
\end{lemma}
\begin{proof}
\cref{inv:rr_rl} is maintained because the $\flush$ operation sets $\rl{z}$ to $\infty$.
Since $\rl{z} > \gamma_{z}$, it is guaranteed by \cref{lem:rayshootinggamma}  that the scan of $\flush$ ends at $q_{z+1}$ or before $q_{z+1}$.
It follows that \cref{inv:TAP-AP} is maintained because $\AP$ does not change and $\flush$ only deletes passive points of $\TP$.
We proceed to prove that \cref{inv:lazyrs} is maintained.
Note that $\rl{z}$ is set to $\infty$ by the end of $\flush$, and that $q_z$ remains the predecessor active point of every $x\in [x_z \ldots x_{z+1}]$, so we need to show $A[x] = \TA[x]$.
Let $x \in [x_z \ldots x_{z+1}]$.
If $x \le x^*$, then before $\flush$ was applied, we had $\TA[x] \ge \ray_{z}(x)$, and therefore by \cref{inv:lazyrs} $A[x] = \min(\TA[x] , \ray_{z}(x)) = \ray_{z}(x)$. Since $\flush$ sets the value of $\TA[x]$ to be $\ray_z(x)$ for $x<x^*$, \cref{inv:lazyrs} still holds.
If $x > x^*$, the value of $\TA[x]$ is not changed by $\flush$.
Since the line $\tell$ between $p_{\mathsf{del}}$ and $p_{\mathsf{end}}$ starts not below the $\ray_{z}$ and ends below $\ray_{z}$, its slope is smaller than $\rl{z}$.
Since the points between $p_{\mathsf{end}}$ and $q_z$ (excluding $q_z$) are passive, the slopes of the corresponding segments are also lower than $\rl{z}$ and therefore $(x,\TA[x])$ is below $\ray_{z}$ for every $x \in (x' \ldots x_{z+1}]$.
Due to \cref{inv:lazyrs} before the application of $\flush$, we have $A[x] = \min(\TA[x], \ray_{z}(x)) = \TA[x]$.
Therefore, assigning $\rl{z} \assign \infty$ and not changing $\TA[x]$ satisfies \cref{inv:lazyrs}.
\end{proof}

\subsection{Implementing the Data Structure}

\subparagraph*{$\Lookup(k)$.} To perform $\Lookup(k)$ we retrieve the predecessor $q_z$ of $k$ in $\TAP$, and invoke $\flush(q_z)$. We then retrieve the predecessor $p_w$ of $k$ in $\TP$ and return $\tell_w(k)$ which is correct by \cref{lem:flush}. All three invariants are clearly maintained by this operation. 

\para{$\AC(i,j,c)$.}
Similar to the implementation in the warmup algorithm (\cref{sec:warmup}), we first perform $\Lookup$ queries to retrieve $A[x]$ for $x\in\{i-1,i,j,j+1\}=\B$. Let $\P_\B$ be the resulting set of $O(1)$ points. We assume that no point of $\P_\B$ was previously in $\TP$ (the other cases are handled similarly). 
We insert the points of $\P_\B$ into $\TP$ and update $\Dbeta$ and $\DAL$ accordingly using $O(1)$ operations. Note that at this point all $O(1)$ mega-segments containing points in $\P_\B$ are flushed (because of the calls to $\Lookup$). 
Next, we apply $\Dbeta.\AtR(i,j,c)$. 
Note that this changes the linear segments between $i-1$ and $i$ and between $j$ and $j+1$.
We update $\DAL$ and $\Dbeta$ to reflect these changes using $O(1)$ additional operations.

Next, 
for every $p_k = (x_k,y_k) \in \P_{\B}$, 
 if $p_k$ just became active, then we insert $x_k$ to $\TAP$  and set the $\rl{}$ value of $x_k$ in $D_\rho$ to be $\infty$. 
Otherwise, if $p_k$ just became passive, then we remove $x_k$ from $\TAP$ and $\DGAM$ if necessary.
Finally, if $\alpha_k=\alpha_{k-1}$, we merge the two segments by removing $p_k$ from $\TP, D_\alpha$ and $D_\beta$.

\begin{lemma}\label{lem:addconst}
    Applying $\AC(i,j,c)$ preserves \cref{inv:rr_rl,inv:lazyrs,inv:TAP-AP}.
\end{lemma}
\begin{proof}
The only points of $\P$ that become active or passive due to $\AC(i,j,c)$ are the points in $\P_{\B}$. Since the mega-segments containing points in $\P_{\B}$ are flushed, \cref{lem:flush} and the explicit handling by $\AC$ of the points of $\P_{\B}$ that become active or passive guarantee that \cref{inv:TAP-AP} is maintained. 

Similarly, the only mega-segments whose $\gamma$ value is changed by $\AC(i,j,c)$ are those containing points of $\P_{\B}$.
The values $\rl{}$ for all these mega-segments are set to $\infty$ either by flushing or explicitly. Hence \cref{inv:rr_rl} holds.

To establish \cref{inv:lazyrs}, 
let $k\in [1\ldots |A|]$ and let $q_z$ be the active point prior to the application of $\AC(i,j,c)$ such that $k\in [x_z \ldots x_{z+1}]$.
\begin{itemize}
    \item If $[x_z \ldots x_{z+1}] \cap [i-1\ldots j + 1] = \emptyset$, then both $A[k]$ and $\TA[k]$  are  not affected by the update.
    Moreover, the predecessor active point of $k$ remains $q_z$ after the update, and $\rl{z}$ was not affected by the update.
    It follows that $\min(\TA[k], \ray_z(k))$ is not changed.
    \item If $[x_z \ldots x_{z+1}] \subseteq [i+1 \ldots j-1]$, then notice that $q_z$ remains active after the update since $\alpha_z$ and $\alpha_{z-1}$ are not affected by the update.
    The value of $\TA[k]$ and the $y$ coordinate of $q_z$ were increased by $c$ via the $\Dbeta.\AtR(i,j,c)$ operation.
    The value $\rl{z}$ was not changed, so $\ray_z(k)$ was increased by $c$ as well.
    It follows that $\min(\TA[k],\ray_z(k))$ was increased by $c$, as required.
    \item If $[x_z \ldots x_{z+1}] \cap \B \neq \emptyset$, then note that we applied a $\flush$ operation on $q_z$, so we have $\rl{z} = \infty$ and $\TA[k] = A[k]$ prior to the application of $\Dbeta.\AtR(i,j,c)$.
    Thus, after the $\flush$ operation, we have $A[k]=\min(\TA[k],\ray_z(k))=\min(\TA[k],\infty)=\TA[k]$.
    After applying the update $\AC(i,j,c)$, the predecessor active point of $k$ is either $q_z$, some point in $\P_{\B}$ (if a point in $\P_{\B}$ became active as a result of the operation), or the predecessor active point of a point in $\P_{\B}$ (if a point in $\P_{\B}$ was $q_z$, and became passive as a result of the update). 
    In all these cases, the predecessor active point $q_a$ of $k$ in the updated representation has $\rl{a} = \infty$ (since either it is a new active point, or it is an existing active point on which a $\flush$ was applied).
    The value of $\TA[k]$ was increased by $c$ via the operation $\Dbeta.\AtR(i,j,c)$. 
    In conclusion, we have $\min(\TA[k],\ray_a(k)) = \min(\TA[k],\infty) = \TA[k]$.
    Since $\TA[k]$ was increased by $c$ if it was necessary, it is now representing the value of $A[k]$ after the $\AC(i,j,c)$ operation.\qedhere 
\end{itemize}

\end{proof}

\para{$\AG(i,j,c)$.}
As was the case in the warmup algorithm, the implementation of $\AG$ is similar to that of $\AC$ except that rather than applying $\AtR(i,j,c)$ to to $D_\beta$, it is applied to $D_\alpha$ to increase the slope of all line segments between $i$ and $j$ by $c$. In the same manner we increase the slope of the corresponding mega-segments by $c$ using $O(1)$ calls to $\AtR$ on $D_\gamma$ (the mega-segments containing $i$ and $j$ need a special treatment since their slope might increase by less than $c$).
Finally, we increase the slope of the pending rays using $\AtR$ on $D_\rho$. 

\begin{lemma}\label{lem:addgrad}
    Applying $\AG(i,j,c)$ preserves \cref{inv:rr_rl,inv:lazyrs,inv:TAP-AP}.
\end{lemma}
\begin{proof}
The proof for \cref{inv:TAP-AP} is identical to that in \cref{lem:addconst}. 
\cref{inv:rr_rl} holds since the only mega-segments whose $\gamma$ and $\rl{}$ change by different values are those containing points of $\P_{\B}$, and those are flushed and handled explicitly by the implementation.

As for \cref{inv:lazyrs},
let $k\in [1\ldots |A|]$ and let $q_z$ be the active point such that $k\in [x_z \ldots x_{z+1}]$ prior to the application of $\AG(i,j,c)$.
\begin{itemize}
    \item If $[x_z \ldots x_{z+1}] \cap [i-1\ldots j + 1] = \emptyset$, 
 then no changes occurs, just like in the proof of \cref{lem:addconst}.
    \item If $[x_z \ldots x_{z+1}] \subseteq [i+1 \ldots j-1]$, notice that $q_z$ is still active since $\alpha_z$ and $\alpha_{z-1}$ were both increased by $c$.
    The value of $\TA[k]$ was increased by $k \cdot c$ via the $\DAL.\AtR(i,j,c)$ operation.
    The $y$ coordinate of $q_z$ was increased by $x_z \ldots c$ via the same operation. 
    The value $\rl{z}$ was increased by $c$ as well, so $\ray_z(k)$ was increased by $x_z + (k-x_z) \cdot c = k\cdot c$.
    It follows that $\min(\TA[k],\ray_z)$ was increased by $k\cdot c$, as required.
    \item If $[x_z \ldots x_{z+1}] \cap B \neq \emptyset$, then $A[k]=\TA[k]$ by the same argument as in the proof of the corresponding case in \cref{lem:addconst}.\qedhere 
\end{itemize}
\end{proof}

\para{$\LLW(i,j,c)$.}
In the description of the algorithm we will say that it \textit{explicitly} performs a ray shooting process from some point $p$. 
By this we mean the following.
First, the mega-segment containing $p$ is flushed.
We then scan the points of $\TP$ starting in $p$ using $\Successor$ queries in $\TP$.
Similarly to the warmup algorithm of \cref{sec:warmup}, we delete the scanned point $p_k=(x_k,y_k)$ from $\TP$ if $y_k \ge \ell(x)$ with $\ell$ being the ray with slope $c$ shot from $p$.
If the scan reaches an active point $q_w = (x_w,y_w)$, and finds that $q_w$ is not below $\ell$ - we perform a $\flush$ operation on $q_w$ before deleting $x_w$ from $\TP, \DAL, D_\beta$.
We also delete $x_w$ from $\TAP, \DGAM, D_\rho$, and update the slopes of the predecessors of $q_w$ in $\TP$ and in $\TAP$ accordingly.
Upon reaching a point $p_k$ that is below the ray $\ell$ (or when reaching $p_b$), we add to $\TP$ the points $p'$ and $p''$ (delete $p_b$ and add $(j,\ell(j))$ if necessary) as described in \cref{sec:warmup}, and the ray shooting process terminates.

We now describe the algorithm for $\LLW(i,j,c)$. (Refer to  \cref{fig:RayShootingAdvanced}).
First, as in the $\AC$ and $\AG$ operations,  we add the points in $\P_{\B}$ to $\TP$, and flush every mega-segment containing a point from $\P_{\B}$ .
If $\alpha_a > c$, we explicitly perform a ray shooting process from $p_a$.
Let $p$ be the point at which the explicit ray shooting process terminated, or $p=p_a$ if $\alpha_a \le c$.

Let $q_{w}=(x_w,y_w)$ be the first active point (weakly) after $p$. 
We use $\DGAM.\nextGT(x_w,c)$ to obtain the  point $q_z = (x_z,y_z)$ from which the next \longrange ray shooting should start.
We then implicitly shoot \local rays in 
every mega-segment starting at an active point $q_t=(x_t,y_t)$ with $x_t\in [x_w \ldots x_z)$. This is done by applying a single  operation, $\Min(x_w,x_{z-1},c)$ on the \addmin data structure $D_\rho$ maintaining the $\rl{z}$ values.
This operation has the effect of setting  $\rl{t} \assign \min(\rl{t},c)$ for all such $t$'s. Next, we explicitly perform the \longrange  ray shooting from $q_z$.

 We keep repeating the above paragraph with the point at which the last explicit \longrange ray shooting process terminated taking the role of $p$.
We stop if an explicit ray shooting reaches $j$ or if $\DGAM.\nextGT(x_w,c)$ returns 'null' or a point beyond $j$. In the latter case, let $q_{b'}$ be the starting point of the mega-segment containing $j$. We  implicitly shoot all the \local rays in all the mega-segments starting not earlier than $p$ and ending no later than $q_{b'}$. Finally, we call $\flush(q_{b'})$, and if $\alpha_{b'} > c$ we explicitly perform a ray shooting process from $q_{b'}$.

To finalize we also need to update the effects around $i$ and $j$. We check for every $p_k\in\P_{\B}$ if $p_k$ is active, and update $\DGAM$ accordingly (similar to this update in $\AC$).

\begin{figure}[htb]
  \begin{center}
 \includegraphics[scale=0.6]{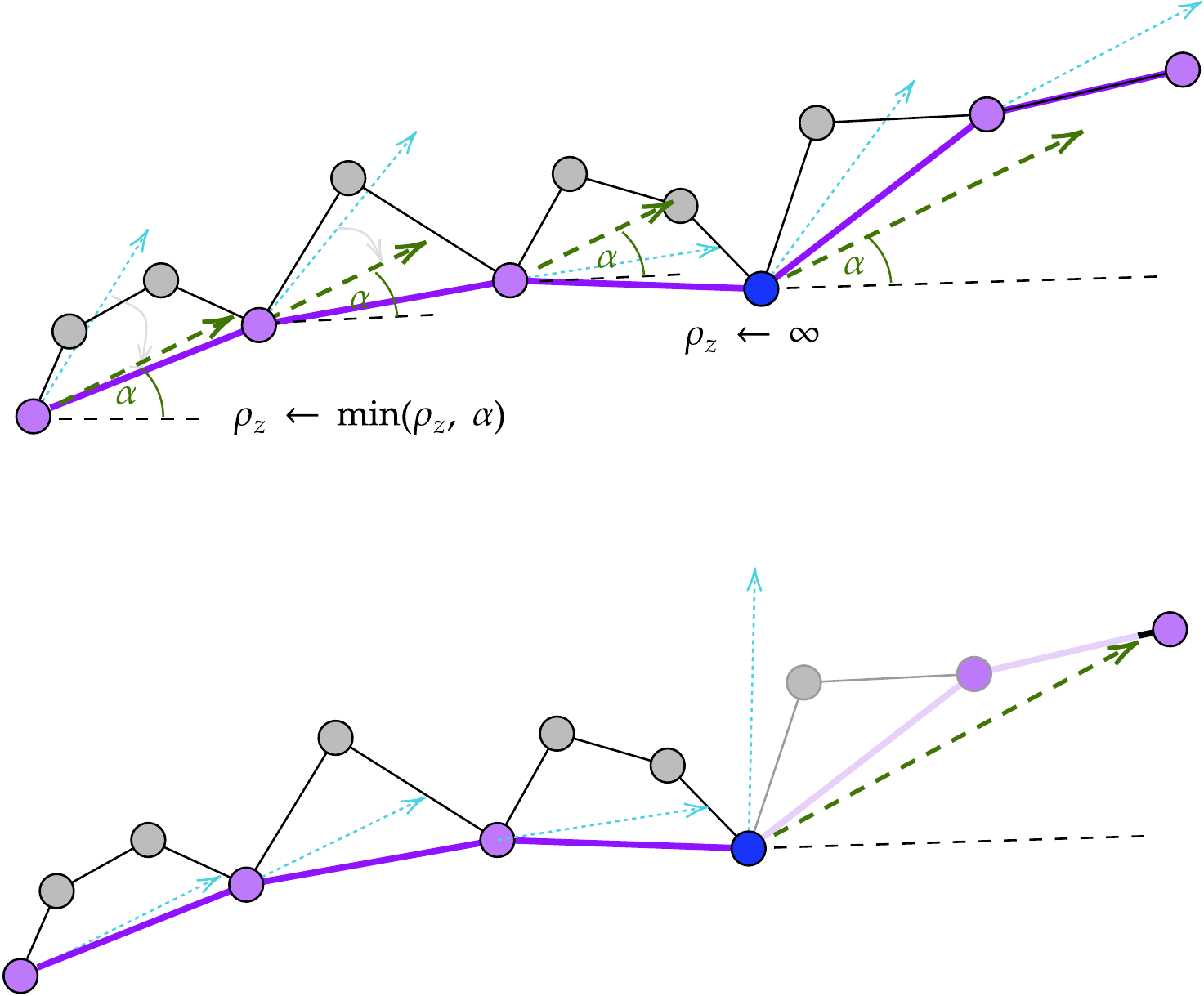} 
  \caption{A demonstration of the first explicit \longrange ray shooting process.
  First, the leftmost active point $q_z$ with $\gamma_z$ larger than $\alpha$ (blue) is identified (assuming that a ray shooting process from $p_a$ is not required).
  Then, the $\rl{}$ values (light blue rays) of the active points preceding $q_z$ is assigned $\rl{} \assign \min( \rl{},\alpha)$.
  Finally, a ray shooting process is explicitly applied from $q_z$.
  Since an explicit ray shooting process starts by applying $\flush$ to the mega-segment of $q_z$, the value of $\rl{z}$ is set to $\infty$.
 \label{fig:RayShootingAdvanced}}
     \end{center}
\end{figure}

\begin{lemma}
     Applying $\LLW(i,j,c)$  preserves \cref{inv:rr_rl,inv:lazyrs,inv:TAP-AP}.
\end{lemma}
\begin{proof}

We start by proving the following claim.
\begin{claim}\label{c:samerayshootings}
The implementation of $\LLW$ performs exactly the same {\em \longrange}ray shooting processes as the warmup algorithm of  \cref{sec:warmup} (same starting points, ending points, and points deleted).
\end{claim}
\begin{proof}
Since the mega-segment containing $p_a$ is flushed, our algorithm shoots a ray from $p_a$ if and only if the warmup algorithm also has.
If we did not shoot a ray from $p_a$, our algorithm shoots a \longrange ray from the first active point $q_z$ with $\gamma_z>c$.
Since $\AP=\TAP$ (\cref{inv:TAP-AP}), the $\gamma$ slopes stored in $\DGAM$ are the slopes of the mega-segments of $\P$.
Therefore, the first \longrange ray shot by the warmup algorithm does not start before $q_z$.
It is possible that the warmup algorithm executed several \local ray shooting processes from active points $q_w$ with $w<z$, but 
by \cref{lem:rayshootinggamma}, the effect of these \local rays is confined to the mega-segment starting at $q_w$, and does not affect any active points.
Therefore, it is guaranteed that the first \longrange ray shot by the warmup algorithm is also from $q_z$.

Consider the first \longrange ray shooting process applied from $q_z$ (or from $p_a$, depending on $\alpha_a$).
Throughout the scan, we always perform a $\flush$ operation on a mega-segment before scanning points in the mega-segment.
Therefore, by \cref{lem:flush}, it is guaranteed that we see the exact same points as in the scan of the warmup algorithm.
It follows that the scan terminates by creating exactly the same point as the warmup algorithm, and deletes all the points between those points and $q_z$).
Our algorithm then proceeds to find where the  next  \longrange ray  should start using $\DGAM.\nextGT$.
The claim follows inductively by repeating the reasoning above on all subsequent \longrange rays.
\end{proof}

\para{\cref{inv:TAP-AP}}

By \cref{c:samerayshootings}, every active point that is deleted from $\AP$ as a result of a ray shooting process is deleted from $\TAP$ as well (and only these points).
By \cref{lem:rayshootinggamma}, the only point that may become active as a result of a ray shooting is $p=(j,A[j]) \in P_{\B}$. This is explicitly handled by the algorithm by flushing the mega-segment containing $j$, which guarantees that the slopes of the segments ending and starting at $p$ are identical in $\P$ and in $\TP$. Hence, $p$ is active in $\TP$ if and only if it is active in $\P$, and \cref{inv:TAP-AP} is satisfied.

\para{\cref{inv:rr_rl}}
Let $q_z$ be an active point in $\TAP$ after the application of $\LLW(i,j,c)$.
Note that since a ray shooting does not create any new active points (except possibly a point $(j,A[j])$), the point $q_z$ was  active also before the application of $\LLW(i,j,c)$ (or $q_z$ is a point in $\P_{B}$).
If $\flush(q_z)$ was invoked then $\rl{z}$ is set to $\infty$ and \cref{inv:rr_rl} is satisfied. 

Otherwise, 
    it must be the case that a \longrange ray shooting process was not executed from $q_z$. 
    Furthermore, no \longrange ray shooting deleted a point in the mega-segment of $q_z$.
    It follows that $q_{z+1}$ was also not affected by the update, since a ray shooting starting in $q_w$ with $w \ge z+1$ does not change the value of $q_{z+1}$, and a \local ray shooting process that may have been applied from $q_z$ does not change $q_{z+1}$ as well.
    Therefore, $\gamma_z$ is unchanged by $\LLW(i,j,c)$.
    Since a \longrange ray did not start at $q_z$, and $q_z$ was not deleted by a \longrange ray shooting, we must have that $\gamma_z <c$.
    Before the application of $\LLW$, we had $\gamma_z < \rl{z}$. Either
    $\rl{z}$ is unchanged by the algorithm, or it was set to $\min(c,\rl{z})$ by a \local ray.
    In both cases, \cref{inv:rr_rl} is maintained.

\para{\cref{inv:lazyrs}}
Let $x\in [n]$.
We distinguish between two cases regarding the mega-segment $[x_z \ldots x_{z+1}]$ containing $x$ prior to the application of $\LLW(i,j,c)$.
\begin{itemize}
    \item If $q_z$ was deleted by an explicit \longrange ray shooting.
    Let $q_w$ be the point from which the \longrange ray started. 
    Hence, after the application of $\LLW(i,j,c)$, $q_w$ is the predecessor active point of $x$, and the mega-segment starting at $q_w$ (and containing $x$) is flushed. Hence the algorithm sets $\rl{w} = \infty$, applies on this segment exactly the same changes as the warmup algorithm $\TA[x] = A[x] = \min(\TA[x], \ray_w(x))$ as required.
    \item If $q_z$ was not deleted by an explicit \longrange ray shooting. 
    It follows from \cref{c:samerayshootings} and from the correctness of the warmup algorithm that if the value of $A[x]$ needs to be modified, it is as a result of a \local ray shooting.
    If the warmup algorithm does not perform a ray shooting process from $q_z$, it must be the case that $\gamma_z  < c$.
    In this case, the assignment $\rl{z} \assign \min(\rl{z},c)$ does not change $\rl{z}$, so $\min(\TA[x],\ray_z)$ is not changed, as required.
    We proceed to treat the case in which a local ray $r$ with slope $c$ is shot from $q_z$ by the warmup algorithm.
    Let $A_{\mathsf{before}}[x]$ and $A_{\mathsf{after}}[x]$ be the values of $A[x]$ before and after $\LLW(i,j,c)$ is applied, respectively.
    The value of $A_{\mathsf{after}}[x]$ is $A_{\mathsf{after}}[x] = \min(A_{\mathsf{before}}[x], r(x))$.
    Let $\ray_{z}^{\mathsf{before}}$ (resp. $\ray_{z}^{\mathsf{after}}$) be the pending ray with slope $\rl{z}$ (resp. $\min(\rl{z},c)$) at $q_z$ before  (resp. after) applying $\LLW(i,j,c)$.
Note that $\ray_{z}^{\mathsf{after}}(x) = \min (\ray_{z}^{\mathsf{before}}(x), r(x))$, so
    \begin{align*}
        & \min(\TA[x] ,\ray_{z}^{\mathsf{after}}(x) ) = \min(\TA[x],\ray_{z}^{\mathsf{before}}(x), r(x)) = 
        \\
        & \min(\min(\TA[x],\ray_{z}^{\mathsf{before}}(x)),r(x)) =
        \min(A_{\mathsf{before}}[x] , r(x)) = A_{\mathsf{after}}[x]
    \end{align*}
    as required.\qedhere
\end{itemize}

\end{proof}

\para{Complexity.}
We start by showing that the number of active points added to $\TAP$ throughout a sequence of $s$ operations is $O(s)$.
This is because $\flush$ operations do not add active points to $\TAP$, and 
each invocation of $\AC$, $\AG$ or $\LLW$ may create $O(1)$ active points.

Consider a sequence of $s$ operations. We use a standard charging argument to prove that the amortized time per operation is $\tilde O(1)$.  
The only difficulty is in charging the time of ray shootings that are performed explicitly in any call to $\flush$ and during $\LLW$, and the time of the implicit \local ray shootings during $\LLW$. 

We charge to each operation the $\tilde O(1)$ time of handling the mega-segments containing the points of $\P_{\B}$, including the time to insert the $O(1)$ new points in $\P_{\B}$ but excluding calls to $\flush$ on these mega-segments. 
Similarly, we charge the $\tilde O(1)$ time update $D_\alpha$ and $D_\beta$ during $\AC$ and $\AG$ to the operation itself.

Each call to $\flush$ may remove  many passive points from $\TP$ and, in addition, takes $\tilde O(1)$ time to insert $O(1)$ new passive points into $\TP$. 
The time to delete each point $p$ of $\TP$ is charged to the insertion of $p$. Calls to $\flush$ on a mega-segments containing points in $\P_{\B}$, or to the mega-segment containing the point $p_a$ in $\LLW$ charge the additional $\tilde O(1)$ required time to the calling operation. All other calls to $\flush$ occur during explicit \longrange ray shootings in $\LLW$, and will be charged next.

Every mega-segment $[q_z,q_{z+1}]$ that is encountered during an explicit \longrange ray shooting, except the last one, results in deleting the active point $q_{z+1}$ from $\TAP$. 
For each such mega-segment we charge $\tilde O(1)$ time for inspecting $q_z$ and calling $\flush(q_z)$ to the the deletion of $q_{z+1}$ from $\TAP$. 
For the last mega-segment, by \cref{lem:activerays}, $q_{z+1}$ is either deleted or becomes inactive, so we charge $\tilde O(1)$ time to $|\TAP|$ decreasing by 1. Note that handling this last segment may include the insertion of $O(1)$ new passive points to $\TP$, which is within the $\tilde O(1)$ charged budget.

Finally, we charge the time of calls to $D_\rho.\Min$ in implementing \local ray shootings implicitly. Each of these calls results from some $\LLW$ operation. 
We charge $O(1)$ such calls to the $\LLW$ operation itself, and each of the remaining calls to the \longrange ray shooting preceding it.

Each operation and each decrease in $|\TAP|$ was charged $\tilde O(1)$ time. Since the sequence consists of $O(s)$ operations and since, as we have shown, $O(s)$ points ever become active, the total time for serving the entire sequence is $\tilde O(s)$.

\para{Handling $\RLW$.} As we had mentioned above, handling $\RLW$ is symmetric to $\LLW$ with the algorithm proceeding right-to-left, starting from $p_b$, and shooting negative rays from active points (the definition of active points remains unchanged). To keep track of pending \local negative rays we maintain an additional Add-min data structure $D'_{\rho}$, and now $A[x]$ is obtained as $\min(\TA[x],\ray_z(x),\ray'_{z+1}(x))$, where $\ray'_{z+1}$ is the pending negative ray going through $q_{z+1}$. The proof of correctness, maintenance of invariants, and analysis of complexity are completely symmetric to those of $\LLW$.   
With that, the proof of \cref{t:mainds} is complete.

\section{The \addmin Data Structure}\label{sec:minadd}
In this section we describe the \addmin data structure of \cref{lem:addminds}.
To explain the main idea of the data structure we assume that no points are added or removed and focus on supporting just $\Min(i,j,c)$ and $\AtR(i,j,c)$. 
We consider the points 
$p_1=(x_1,y_1), p_2=(x_2,y_2), \ldots$ ordered by their $x$-coordinates.
The \addmin data structure is recursive. It consists of an Interval-add data structure $D$, and of a recursive instance of \addmin $R$, which only stores a constant fraction (2/3) of the points.

The points are partitioned into pairs of consecutive points. 
The {\em representative} of a point $p=(x,y)$ is the first point in the pair that $p$ belongs to. Let $M(x)$ denote the $x$-coordinate of the representative of $p$. 
Initially $D$ stores all the points $p_k=(x_k,y_k)$, and $R$ stores only the representatives, which are initialized to $\infty$. Namely, $(x_1,\infty), (x_3,\infty), (x_5,\infty), \ldots$.
We maintain the invariant that $y_k = \min(D.\Lookup(x_k), R.\Lookup(M(x))$, so a $\Lookup(x_k)$ query on the data structure can be served with a single $\Lookup$ on $D$ and a single recursive $\Lookup$ on $R$, which would take total polylogarithmic time.

We use a service operation $\AssignOp(x_k,c)$, that assigns $y_k \assign c$ if there is a point $p=(x_k,y_k) \in S$.
The operation $\AssignOp(x_k,c)$ is implemented by making $D$ store the value $c$ as follows (note that doing $\AssignOp$ on $D$ is trivial using $\Remove$ and $\Insert$). Let $p_{k'}$ be the other point in the pair with $p_k$. We obtain the value of $y_{k'}$ using a $\Lookup$ operation, and call $D.\AssignOp(x_{k},c)$, $D.\AssignOp(x_{k'},y_{k'})$, and recursively call $R.\AssignOp(M(x_k),\infty)$. 

To implement $\Min(i,j,c)$, let $p_a$ (resp. $p_b$) be the first point with $x_a \geq i$ (resp. $x_b \leq j$).  Assume first that $p_a$ is a representative (i.e., $M(x_a) = x_a$) and $p_b$ is not, so the effected range $[i,j]$ exactly corresponds to a range of consecutive pairs of points. We simply invoke $R.\Min(i,j,c)$, which clearly correctly implements the $\Min$ operation while maintaining the invariant. 
If $p_a$ is not a representative then we need to handle the pair containing $p_a$ differently since we do not want to affect the value of $p_{a-1}$. We obtain the values of $y_{a-1}$ and $y_a$ before the update using two $\Lookup$ operations, and assign these values by calling $D.\AssignOp(x_{a-1},y_{a-1})$, $D.\AssignOp(x_a,\min(y_a,c))$,  $R.\AssignOp(x_{a-1},\infty)$ (this last call is a recursive assignment). This maintains the invariant and guarantees that both $p_{a-1}$ and $p_a$ are correctly represented. We handle similarly the case where $p_b$ is a representative.

We implement $\Add(i,j,c)$ in a similar spirit. If $p_a$ is a representative and $p_b$ is not, we simply invoke $D.\AtR(i,j,c)$ and $R.\AtR(i,j,c)$. Otherwise, we handle the endpoints of the intervals explicitly in the manner described for $\Min$. 

To support insertions and deletions of points we can no longer work with the rigid partition into consecutive pairs. Instead, we shall use the standard technique of partitioning the points into segments consisting of a single point or two consecutive points. Only the first point from each segment is represented in the recursive structure. 
We make sure to merge consecutive segments whenever both contain just a single point. This guarantees that the number of segments is at most $2/3$ the number of points, and hence the recursive structure is sufficiently small. 

To keep track of the partition into segments we maintain the representatives in a predecessor data structure $M$. The representative of a point $p_k$ is then given by the predecessor of $x_k$ in $M$. The invariant now becomes $y_k = \min(D.\Lookup(x_k), R.\Lookup(M.\Predecessor(x_k))$.
We denote by $D[x]$ (resp. $R[x]$) the value of the $y$ coordinate of the point with $x$ coordinate in $D$ (resp. in $R$) if such a point exists. We denote by $M[x]$ the value of $M.\Predecessor(x)$.
With this notation the invariants we maintain is:
\begin{invariant}\label{inv:addminvalue}
For every $p = (x,y)$ in the data structure, we have  $y = \min(D[x],R[M[x]])$.
\end{invariant}

We denote the segments by $s_1,s_2 \ldots s_r$. The invariant on the segments is:
\begin{invariant}\label{inv:addminsegments}
    For every $i\in[1\ldots r-1]$ either $s_i$ is of length two or $s_{i+1}$ is of length two.
\end{invariant}
Note that a direct consequence of \cref{inv:addminsegments} is that $r \le \ceil{\frac{2n}{3}}$.
Also note that, for every $i\in [1\ldots |D|]$, $M[x_i]$ is either $x_i$ or $x_{i-1}$. 

For completeness we give all the details of the data structure.
Upon initialization, the data structure contains no elements and $D,M,$ and $R$ are empty.

\para{$\Lookup(x)$.}
We perform a $\Lookup(x)$ query on $D$, and a (recursive) $\Lookup(M[x])$ query on $R$ and return the minimum of the two.

\para{$\AssignOp(x,y)$.}
Let $x_{i'}=M[x_i]$ and let $s_a$ be the segment containing the point $p$ with $x$ value of $x$.
We assign $D[x] \assign y$. If there is another point $p'=(x',y')$ in $s_a$ we also  assign $D[x']\assign y'$.
Note that we can identify the two candidates for being $p'$ using predecessor and successor queries, and confirm which is in $s_a$ using $M$.
We conclude by (recursively) assigning $R[x_{i'}]\assign \infty$.

\para{$\ShiftOp(x,x')$.}
We introduce a service operation that would be useful for maintaining the invariants throughout insertions and deletions.
The input for $\ShiftOp$ is an $x$ value of a point $p_i = (x,y_i)$ in $D$, and a new value $x'$ satisfying $x_{i-1} \le  x' \le x_{i+1}$.
The operation $\ShiftOp$ replaces $p_i$ with $p = (x',y_i)$.
Note that due to the constraint on $x'$, the new point $p$ can enter the segment from which $p_i$ is removed.
We implement $\ShiftOp(x,x')$ as follows.
First, we remove $p_i$ from $D$ and insert $p = (x',y_i)$ instead (a $\Lookup$ operation is required to acquire $y_i$).
If $M[x_i] = x_i$, we also remove $x_i$ from $M$ and add $x'$ instead.
Finally, if $x_i$ was replaced in $M$ we also recursively apply $R.\ShiftOp(x,x')$.

\para{$\Insert(x,y)$.}
Let $x_p$ and $x_s$ be the predecessor and the successor $x$ values of $x$ in the data structure, respectively.
Let $x_a = M[x_p]$ and $x_b= M[x_s]$.
\begin{itemize}
    \item if $x_a\ne x_b$ 
    we apply $D.\Insert(x,y)$.
    Then, we update $M$ and $R$ as follows:
    \begin{itemize}
        \item If the segments containing $x_s$ and $x_p$ are both of length two, 
        then we create a new segment $s=[(x,y)]$ and apply $R.\Insert(x,\infty)$ (recursively).
        We also add $x$ to $M$.
        \item If one of the segments containing $x_p$ and $x_s$ are of length one, 
        assume that the segment $s$ containing $x_p$ is the segment of length one.
        We update the value of $D[x_{p}]$ to be the proper value of $x_p$ in our data structure by applying $D[x_{p}]\assign \Lookup(x_{p})$. 
        Moreover, we apply $R[x_a]\assign\infty$ to guarantee \cref{inv:addminvalue}.
        Note that in this operation, we add a point to a segment.
        If in this process the added point $(x,y)$ becomes the first point of the segment previously containing a single point $p'= (x',y')$, we need to update $M$ s.t. $x'$ is mapped to $x$ and update $R$ to contain a point with $x$ coordinate $x'$ instead of $x$.
        This is achieved by replacing $x$ with $x'$ in $M$ and applying $R.\ShiftOp(x,x')$.
    \end{itemize}
    \item If $x_a=x_b$, then inserting some point $(x',y')$ right before the segment $s$ containing $x_p$ is a case that we already covered.
    Thus, instead of inserting $(x,y)$ to $s$, we replace $x_p$ with $(x,y)$ and insert $x_p$ before $s$ following the previous cases.
    Let $p_p = (x_p,y_p)$ (obtained via $\Lookup(x_p)$). 
    We replace $p_p$ with $p=(x,y)$ by removing the point with $x$ value $x_p$ from $D$ and insert $(x,y)$ instead.
    We also remove $x_p$ from $M$ and add $x$ instead.
    Recall that $p_s = (x_s,y_s)$ is also in the segment 
    and we assign $\Lookup(x_{s})$ to $D[x_s]$ and $\infty$ to $\R[x_a]$.
    We also apply $R.\ShiftOp(x_p,x)$.
    With that, we have replaced $p_p$ with $p=(x,y)$.
    We proceed to insert $(x_p,y_p)$ via one of the previous cases.
\end{itemize}

\para{$\Remove(x)$.}
Let $x'=M[x]$. Let $s_{i}$ be the segment containing $x$. 
\begin{itemize}
    \item If both $s_{i-1}$ and $s_{i+1}$ are of length two or do not exist, then we remove $x$ from $D$ by applying $D.\Remove(x)$.
        If $s_{i}$ is of length two, then it may be the case that the first point in $s_{i}$ was removed and is now $p'=(x',y')$.
        If it is the case, we replace $x$ with $x'$ in $M$ and apply $R.\ShiftOp(x,x')$.
        If $s_{i}$ is of length one, the segment $s$ should be removed.
            We remove $R[x']$ by applying $R.\Remove(x')$.
            We also remove $x'$ from $M$.
    
    \item If $s_{i-1}$ or $s_{i+1}$ is of length one (assume that $s_{i-1}$ is of length one).
        Notice that by \cref{inv:addminsegments} $s_{i}$ is of length two so $s_{i-1}\cup s_{i}$ has exactly $3$ points.
        Let $p'$ be the point in $s_{i-1}$ and $p''$ be the other point in $s_{i}$ other than $(x,y)$.
        We remove the point $p'=(x',y')$ from $s_{i-1}$ as described in previous cases.
        We then manipulate the $x$ and $y$ values of the points in $s_i$ via operations on $D$ and assign and shift operations on $R$ (as described  $\Insert$) to replace them with $p'$ and $p''$.
        Notice that we already described how to remove elements in a segment of length one (due to \cref{inv:addminsegments} it must be the case that $s_{i-2}$ is of length two, if it exists).
        If the first point in $s_{i}$ is changed as a result of the deletion, we update $M$ and $R$ accordingly as in $\Insert$.
\end{itemize}

\para{$\Add(i,j,c)$.}
Let $p_p = (x_p,y_p)$ and $p_s = (x_s,y_s)$ be the $x$ successor of $i$ and the $x$ predecessor of $j$ in the data structure, respectively.
Let $x_a= M[x_p]$ and $x_b=M[x_s]$. Let $s_{a'}$ and $s_{b'}$ be the segments containing $x_p$ and $x_s$, respectively.
First, we update the values of all (at most 4) elements in $s_{a'}$ and $s_{b'}$ that their value need to be changed. 
This is done via $\AssignOp(x, \Lookup(x)+c)$ operation for every point $p=(x,y)$ in $s_{a'}$ or in $s_{b'}$ with $x\in [i\ldots j]$.
Let $p_1 = (x_1,y_1)$ be the first point in $s_{a'+1}$, and $p_2= (x_2,y_2)$ be the last point in $s_{b'-1}$, it remains to add $c$ to all elements in the range $[x_1\ldots x_2]$.
Since the $y$ value of every point $(x,y)$ is represented by $\min(D[x],R[M[x]])$, we add $c$ to the two parts of the representation as follows.
We add $c$ to the $y$ value for every point $p=(x,y)$ of $D$ with $x \in [x_1 \ldots x_2]$.
In addition, we (recursively) add $c$ to the corresponding segments via a $R.\Add(M[x_1],M[x_2],c)$ operation.

\para{$\Min(i,j,c)$.}
As before, let $p_p = (x_p,y_p)$ and $p_s = (x_s,y_s)$ be the $x$ successor of $i$ and the $x$ predecessor of $j$ in the data structure, respectively.
Let $x_a= M[x_p]$ and $x_b=M[x_s]$. Let $s_{a'}$ and $s_{b'}$ be the segments containing $x_p$ and $x_s$, respectively.
First, we update the values of all (at most 4) points in $s_{a'}$ and $s_{b'}$ that their value need to be changed. 
This is done via $\AssignOp(\min(x,\Lookup(x)),x)$ for every point $(x,y)$ in $s_{a'}$ and $s_{b'}$ with $x\in[i\ldots j]$.
Let $p_1=(x_1,y_1)$ be the first point in $s_{a'+1}$, and $p_2 = (x_2,y_2)$ be the last element in $s_{b'-1}$, it remains to apply the $\Min$ operation to all the points $(x,y)$ with $x \in [x_1\ldots x_2]$.
Since the representation of the $y$ value of a point $p=(x,y)$ is $y = \min(D[x],R[M[x]])$, it is sufficient to apply $R.\Min(M[x_1],M[x_2],c)$.
This is due to the equation $\min(\min(x,y),z)=\min(x,\min(y,z))$.

\para{Complexity.}
A $\Lookup$ performs $O(1)$ operations on Interval-add and predecessor data structures and a recursive call to $R$. 
Thus, $T_{\Lookup}(n)=T_{\Lookup}(\ceil{\frac{2n}{3}})+O(\log n)$, and therefore the time complexity of $\Lookup$ operation is $O(\log^2 n)$.

An $\AssignOp$ (resp. $\ShiftOp$) operation is performed using $O(1)$ operations on Interval-add and predecessor data structures, a constant number of $\Lookup$ operations and a recursive call of $\AssignOp$ (resp. $\ShiftOp$) to $R$. 
Thus, $T_{\AssignOp}(n) =T_{\AssignOp}(\ceil{\frac{2n}{3}})+O(\log^2 n)$, and therefore the time complexity of $\AssignOp$ (resp. $\ShiftOp$) is $O(\log^3 n)$.

An operation $\Insert$, $\Remove$, $\Add$ or $\Min$ is performed by applying $O(1)$ operations on interval-add and predecessor data structures, a constant number of $\Lookup$, $\AssignOp$ and $\ShiftOp$ operations and a recursive call to $R$. 
Thus, for all these operations $T(n)=T(\ceil{\frac{2n}{3}})+O(\log^3 n)$, and therefore their time complexity is $O(\log^4 n)$.\footnote{We did not attempt to optimize the degree of the polylog. Some of the $\log$ factors can be easily avoided.}

\appendix
\section{Missing proofs}\label{sec:appendix}

\para{Proof of 
\cref{lem:onlydiagonal}.}
We prove the case where $(x,y),(x,y+1) \in B$, the other two cases are similar.  
All edges entering $(x+1,y+1)$ have the same weight $c$. 
We claim that the values $\dist(x,y)$ are weakly monotone along every row and column in a block.
This implies that $\dist(x,y)+c$ (the path to $(x+1,y+1)$ through $(x,y))$) is not larger than $\dist(x,y+1)+c$ (the path to $(x+1,y+1)$ through $(x,y+1))$).  

To see why the values $\dist(x,y)$ are weakly monotone along every row and column in a block, consider two vertices $(x,y),(x,y+1)$ in the same block $B$. We show that $\dist(x,y) \le  \dist(x,y+1)$ (a  symmetric proof shows that $\dist(x,y) \le  \dist(x+1,y)$).
Let $P$ be a shortest path from $(0,0)$ to $(x,y+1)$. 
Let $(x',y)$ be the last vertex in $P$ with second coordinate $y$. 
If $x'=x$ then clearly $\dist(x,y) \le \dist(x,y+1)$. Otherwise, $x' < x$. 
Then, we can assume that the suffix of $P$ starting from $(x',y)$ is composed of a single diagonal edge followed by zero or more vertical edges (since a horizontal edge followed by a vertical edge is always not shorter than just using the diagonal edge). 
Now consider the path $P'$ from $(0,0)$ to $(x,y)$ that is identical to $P$ until $(x',y)$ and from $(x',y)$ continues vertically. Paths $P'$ and $P$ only differ in the suffix from $(x',y)$. But in this suffix they both use the same number of edges $x$-$x'$ and the same edge weights (since all edges in a block have the same weight). 
This means that $P'$ and $P$ have the same length and thus $\dist(x,y) \le \dist(x,y+1)$).      

\begin{lemma}\label{lem:ray}
For every $x\in [i\ldots j]$ the procedure $\LLW(i,j,\alpha)$ described in \cref{sec:warmup} assigns  $A[x]\assign \min_{t\le x}(A[t] + (x-t)\alpha)$.
\end{lemma}
\begin{proof}
Let $L(k)$ be the value assigned to $A[k]$ by $\LLW(i,j,\alpha)$, we first claim that: 

$$L(k) = \begin{cases}
A[i] & k=i
\\
\min(A[k], L(k-1) + \alpha) & k \in [i+1 \ldots j]
\end{cases}
$$
We prove this by induction on $k-i$.
For $k-i = 0$, $\LLW(i,j,\alpha)$ assigns $A[i] \leftarrow \min_{t\in [i\ldots i]}(A[t] + (i-t) \alpha) = A[i] = L(i)$ as required.
For $k - i > 0$:
\begin{align*}
A[k] & \leftarrow \min_{t\in [i\ldots k]}(A[t] + (k-t) \alpha) = \min\big( \min_{t\in [i\ldots k-1]}(A[t] + (k-t) \alpha) , A[k] \big) \\
&=   \min\big( \min_{t\in [i\ldots k-1]}(A[t] + (k- 1 - t) \alpha) + \alpha , A[k] \big) = \min(L(k-1) + \alpha , A[k] ) = L(k). 
\end{align*}

By the above, we need to prove that after $\LLW$ is applied, $A[x] = L[x]$ for every $x\in[i\ldots j]$ (clearly, the operation $\LLW$ does not change $A[x]$ for $x\notin [i\ldots j]$).
We prove this claim by induction on $x \in [i \ldots j]$. 
For $x=i$ the claim holds since $L[i]=A[i]$ and indeed, $\LLW$ does not change $A[i]$. 
Otherwise, assume that the claim holds for $x-1 \in [i\ldots j-1]$.
Let $z\in[q\ldots b]$ be the maximal value such that a ray shooting process started from $p_z = (x_z,y_z)$ and $x_z < x$. 
Let $w\in [z+1\ldots |\P|$ be minimal index such that $p_w$ is below the ray $\ray_z$ with slope $\alpha$ shot from $p_z$. Let $p'=(x',y')$ be the intersection point of $\ell_{w-1}$ and the ray with slope $\alpha$ shot from $p_z$.
We distinguish between two cases.

\begin{itemize}
    \item \textbf{Case 1: $x' \ge x-1$.}
    In this case, the linear segment containing $x-1$ after $\LLW(i,j,\alpha)$ is applied is a sub-segment of the ray $\ray_z$. 
    Therefore, the value assigned to $A[x-1]$ by the procedure is $\ray_z(x-1)$.
    According to the induction hypothesis, we have $\ray_z(x-1) = L[x-1]$.
    We consider two cases regarding the value $A[x]$ before the update is applied.
    \begin{itemize}
        \item \textbf{$A[x] \ge L[x-1] + \alpha$}. In this case, the point $(x,A[x])$ is not below the ray $\ray_z$ and $p'$ must be to the left of $(x,A[x])$.
        It follows that after the procedure is applied, $x$ is also on a linear segment that is a sub-segment of $\ray_z$ and therefore $A[x] = \ray_z(x)$. 
        Indeed, in this case $L[x] = \Min(L[x-1] + \alpha,A[x]) = L[x-1] + \alpha = \ray_z(x-1) + \alpha = \ray_z(x)$ as required.
        \item \textbf{$A[x] < L[x-1] + \alpha$} Note that in this case, we have $p_w = (x,A[x])$.
        Recall that when $p_w$ is met in the ray shooting process, a segment connecting $(x-1,\ray_z(x-1))$ and $(x,A[x])$ is created.
        Therefore, the value of $A[x]$ is not changed by the ray shooting process.
        Indeed, in this case we have $L[x] = \min(L[x-1] + \alpha,A[x]) = A[x]$.
    \end{itemize}
    \item \textbf{Case 2: $x' < x-1$.}
    In this case, the procedure $\LLW(i,j,\alpha)$ did not change the value of $x-1$ and we have $L[x-1] = A[x-1]$ from the induction hypothesis.
    Specifically, the ray shooting process from $q_z$ terminated, and the next ray shooting process, if a necessary one exists, is from a point $p_{z'}$ with $x_{z'} > x-1$.
    This implies that the slope of the linear segment containing $x-1$ is at most $\alpha$.
    Therefore, we must have $A[x]< A[x-1] + \alpha =L[x-1] + \alpha$ and as a result $L[x] = A[x]$.
    Whether or not a ray shooting starts from $(x,A[x])$, the value of $A[x]$ is not changed by $\LLW(i,j,\alpha)$.
    If a ray shooting does not start from $(x,A[x])$ - the linear segment containing $x$ is not affected by $\LLW(i.j.\alpha)$ as no ray shooting process interacted with it.
    If a ray shooting process starts from $(x,A[x])$, the linear segment containing $x$ after the update is applied will be a sub-segment of a ray $\ell^*$ starting from $(x,A[x])$, and clearly $\ell^*(x) = A[x]$ as required.\qedhere 
\end{itemize}
\end{proof}

\para{Proof for \cref{lem:activeraysright}.}
Assume to the contrary that a ray shooting process starts at a passive point $p_z \neq p_b$.
 If $p_z$ is the first point from the right where a ray shooting starts, then $z$ is the maximal index in $[a\ldots b]$ with $\alpha_{z-1} <- \alpha$. But since $p_z$ is passive, we have $-\alpha > \alpha_{z-1} \ge \alpha_{z}$, contradicting the maximality of $z$ (note that $p_{z} \neq p_b$ so $z+1 \in [a\ldots b]$).

Otherwise, let $p_q$ be the last point before $p_z$ from which a ray shooting process occurred.
Let $p_{q'}$ be the first point below the ray shot from $p_q$. Since $p_z$ is the next point from which a ray is shot, $z$ is the rightmost point in $[a \ldots q']$ with $\alpha_{z-1} < -\alpha$.
Since $p_z$ is passive, we have $-\alpha > \alpha_{z-1} \ge \alpha_{z}$.
If $z \neq q'$, we have $z+1 \in [a\ldots q']$, a contradiction to the maximality of $z$.
Otherwise, $p_z = p_{q'}$ is the first point below the ray with slope $-\alpha$ shot from $p_q$. 
It follows from $\alpha_z < -\alpha$ that $p_{z+1}$ is below the ray as well, and a contradiction to $p_z = p_{q'}$ being the first point below the ray.
\end{document}